\documentclass[11 pt]{article}
\usepackage{amsmath, mathtools}
\usepackage{amsthm}
\usepackage{amssymb}
\usepackage{graphicx}
\usepackage{appendix}
\usepackage{color}
\usepackage{float}
\usepackage{hyperref}
\usepackage{bbm}
\usepackage{easyReview}
\usepackage{multirow}

\newtheorem{definition}{Definition}[section]
\newtheorem{assumption}[definition]{Assumption}
\newtheorem{theorem}[definition]{Theorem}
\newtheorem{lemma}[definition]{Lemma}

\numberwithin{equation}{section}

\DeclareMathOperator*{\esssup}{ess\,sup}

\begin{document}
\title{ Analysis of optimal portfolio on finite and small-time horizons for a stochastic volatility model with multiple correlated assets}

\author{ Minglian Lin\footnote{Department of Mathematics, North Dakota State University, Email: minglian.lin@ndsu.edu} \quad and \quad Indranil SenGupta\footnote{Department of Mathematics and Statistics, City University of New York (CUNY)- Hunter College, Email: indranil.sengupta@hunter.cuny.edu}}

\date{\today}

\maketitle
\begin{abstract}

In this paper, we consider the portfolio optimization problem in a financial market where the underlying stochastic volatility model is driven by $ n $-dimensional Brownian motions. At first, we derive a Hamilton-Jacobi-Bellman equation including the correlations among the standard Brownian motions. We use an approximation method for the optimization of portfolios. With such approximation, the value function is analyzed using the first-order terms of expansion of the utility function in the powers of time to the horizon. The error of this approximation is controlled using the second-order terms of expansion of the utility function. It is also shown that the one-dimensional version of this analysis corresponds to a known result in the literature.
We also generate a close-to-optimal portfolio near the time to horizon using the first-order approximation of the utility function. It is shown that the error is controlled by the square of the time to the horizon. Finally, we provide an approximation scheme to the value function for all times and generate a close-to-optimal portfolio.

\end{abstract}



	
\section{Introduction}

One of the most fundamental questions in finance is related to portfolio optimization. If we start with a given amount of initial capital, we can ask the questions related to the capital withdrawal streams that we can fund into the future. Moreover, among those streams, how do we select the ``best" one, according to some given criteria of optimality? In other words, if the wealth of a portfolio is $x$ at an initial time $t$, the usual goal is to invest in such a way that maximizes the expected utility of wealth at the time horizon $T$. Mathematically, if $U_T$ is a function modeling the investor's utility of wealth at time $T$, then the goal is typically to find a portfolio $\pi$, that maximizes $\mathbb{E}(U_T| \mathcal{F}_t)$, where $\mathcal{F}_t$ represents the $\sigma$-algebra of information up to the initial time $t$. One of the first major works of portfolio theory (see \cite{Merton}), examines the combined problem of optimal portfolio selection and consumption rules for an individual in a continuous-time model where the investor's income is generated by stochastic returns on assets. The paper derives the optimality equations for a multi-asset problem when the rate of returns is generated by a standard Brownian motion.

In the work \cite{Kumar} this portfolio-optimization problem is considered in a simple incomplete market and under a general utility function, and a closed-form formula for a trading strategy is obtained. It is shown that the formula approximates the optimal trading strategy. The problem of portfolio optimization is particularly interesting if the corresponding process is highly volatile. This is accomplished in our earlier work \cite{LinSG}, where we obtain a closed-form formula for an approximation to the optimal portfolio in a small-time horizon, where the asset price incorporates stochastic volatility with jumps. we consider the portfolio optimization problem in a financial market under a general utility function. We obtain a closed-form formula for an approximation to the optimal portfolio in a small-time horizon. This is obtained by finding the associated Hamilton-Jacobi-Bellman (HJB) integro-differential equation. In \cite{LinSG}, we also prove the accuracy of the approximation formulas. This work is further developed in \cite{LinSG2}, where VaR is estimated for a diversified portfolio consisting of multiple cash commodity positions driven by models similar to \cite{LinSG}.

It is also worth noting that for the case of the infinite-horizon, the optimal consumption problem with a path-dependent reference under exponential utility is studied in \cite{FinSt1}. It is shown that for this case, the HJB equation can be heuristically expressed in a piecewise manner across different regions to take into account different constraints. In \cite{FinSt3} a continuous-time stochastic optimization problem with infinite horizon is considered that includes as a particular case portfolio selection problem under transaction costs for models of stock and currency markets. With a relevant geometric formalism, the authors show that the Bellman function is the unique viscosity solution of the corresponding HJB equation. On the other hand, in the high-frequency limit, the paper \cite{FinSt2} finds an explicit formula for locally mean-variance optimal strategies. In the limit, it is shown that conditionally expected increments of fractional Brownian motion converge to a white noise, shedding their dependence on the path history and the forecasting horizon and making dynamic optimization problems tractable. Consequently, it is assumed that the asset price is driven by a fractional Brownian motion. It is worth noting that fractional Brownian motions, in the recent literature, have been successfully implemented for various financial products such as variance and volatility swaps, and hedging (see \cite{Nick}).

The works described thus far consider a single asset portfolio. However, it is well-known that investment on different assets can provide best opportunity to increase the value of portfolio consistently and over the long term. A major benefit of multi-asset investing is its underlying diversification. By investing in more than one kind of asset, one can spread risk and improve the potential for returns. It is known that a multi-asset portfolio is flexible to respond to changing market conditions. Such portfolio finds out areas of greater potential return while attempting to avoid sectors that could add unnecessary risk. As a result, a multi-asset investment approach aims to provide smoother positive returns over time and it dynamically adjusts exposure to take advantage of short-term opportunities and/or to avoid unnecessary risk.

In \cite{Multi2, Multi1}, the authors provide (and generalize) a computational study of the problem of optimally allocating wealth among multiple stocks and a bank account, to maximize the infinite horizon discounted utility of consumption. The model studied in that work allows for correlation between the price processes, which in turn gives rise to interesting hedging strategies. The authors provide a numerical method for dealing with this problem and present computational results that describe the impact of volatility, risk aversion of the investor, level of transaction costs, and correlation among the risky assets on the structure of the optimal policy. 
In the paper \cite{Kumar}, the authors consider the corresponding HJB partial differential equation, where the associated risky process is assumed to be driven by only one Brownian motion. For the small-time horizon, the approximate value function is obtained by constructing classical sub-solution and super-solution to the HJB partial differential equation using a formal expansion in powers of horizon time. It is worth noting that the method of super-solution and sub-solution to bind the solution of interest has been recently used to optimize portfolios in the commodity market (see \cite{Roberts1, Roberts2}).

In this paper, we obtain a closed-form formula for an approximation to the optimal portfolio in a small-time horizon, where the underlying stochastic volatility model is driven by an $ n $-dimensional Brownian motion. At first, we derive an HJB equation including the correlations among the standard Brownian motions.
 Similar to the observation of \cite{Kumar}, we note that the well-posedness of the associated HJB equation is not established. Consequently, we do not assume the value function is a classical solution of the HJB equation. (It is also worth noting that without assuming the value function satisfies the HJB equation, in \cite{Nadtochiy}, it is proved that the integral of the viscosity solution of the marginal HJB equation is indeed the value function.)

For the optimization of portfolios, we choose the Jacobi method in order to apply the smallest number of approximations. Then we approximate the value function using the first-order terms of expansion of the utility function in the powers of time $ (t) $ to the horizon $ (T) $ (i.e.,  in the powers of $ T - t $). We control the error of this approximation using the second-order terms of expansion of the utility function. After dimension reduction, it is shown that the approximation of value function agrees with the primary result in \cite{Kumar}. This demonstrates that the proposed  $ n $-dimensional model in this paper includes all of the advantages of the corresponding $ 1 $-dimensional model in the literature. We also generate a close-to-optimal portfolio near the time to horizon $ (T - t) $ using the first-order approximation of the utility function. The error is controlled by the square of the time to horizon $ (T - t)^2 $. Finally, we provide an approximation scheme to the value function for all times $ t \in  [0, T] $ and generate the close-to-optimal portfolio on $ [0, T] $.

The rest of the paper is organized as follows. In Section \ref{sec2}, we provide the detailed description of the underlying model and assumptions. In Section \ref{sec3}, we provide the main result of this paper. We also provide a result for a special case of the main theorem. We show that this special case agrees with related results in the literature. In Section \ref{sec4}, we provide an approximation result of the main result that is described in Section \ref{sec3}. In addition, we briefly outline the portfolio optimization on a finite time horizon. Numerical results are shown in Section \ref{sec5}. Relevant tables and figures for numerical examples are provided in this section. Finally, a brief conclusion is provided in Section \ref{sec6}. Proofs of some technical lemmas are provided in Appendix \ref{apendixA}. The proof of the main theorem is provided in Appendix \ref{apendixA1}. A justification of a major assumption is provided in Appendix \ref{apendixC}.

\section{Market model and underlying assumptions}
\label{sec2}

	We consider the following stochastic volatility market model, that is driven by $ n $-dimensional Brownian motions. 
	We suppose that the investor will be investing in \emph{one} risk-free asset and $ n $ risky assets. Without loss of generality, we let the unit price of risk-free asset be $1$, and define the unit price of risky assets $ \boldsymbol{S}(t) = (S_1(t), \cdots, S_n(t))^\top $ by the following processes:
	\begin{align} 
			dS_i(t) = S_i(t) \Big[ \mu_i \big(Y(t)\big)dt + \sigma_i \big(Y(t)\big)dW_i(t) \Big], \quad S_i(0) > 0, \quad i = 1, \dots, n,
\label{koekoe}
	\end{align}
where $W_i$, $i=1,\dots,n$ are standard Brownian motions.  This can be written as
	\begin{align} \label{St}
			d\boldsymbol{S}(t) = \boldsymbol{S}(t) \circ \Big[ \boldsymbol{\mu} \big(Y(t)\big)dt + \boldsymbol{\sigma} \big(Y(t)\big)\circ d\boldsymbol{W}(t) \Big],
	\end{align}
 where $ \circ $ is Hadamard product, $ t $ is time, $ \boldsymbol{\mu}(\cdot) = (\mu_1(\cdot), \cdots, \mu_n(\cdot))^\top $ is the mean rate of return, $ \boldsymbol{\sigma}(\cdot) = (\sigma_1(\cdot), \cdots, \sigma_n(\cdot))^\top $ is the volatility, and $ \boldsymbol{W}(t) = (W_1(t), \cdots, W_n(t))^\top $.
	Here we consider that $ \boldsymbol{\mu}(\cdot) $ and $ \boldsymbol{\sigma}(\cdot) $ are the functions of the same stochastic factor $ Y(t) $. With the insight from \cite{Nadtochiy} (equation (2) in \cite{Nadtochiy}), we define the stochastic factor $ Y(t) $ by
	\begin{align} \label{Yt}
		dY(t) = b\big(Y(t)\big)dt 
		+ a\big(Y(t)\big) \left[ \sum_{i=1}^{n}\omega_i dW_i(t) + \bigg[1 - \sum_{i=1}^{n} \omega_i^2\bigg]^{\frac{1}{2}} dW_0(t) \right], 
		\quad  Y(0) \in \mathbb{R},
	\end{align}
	where $ a $ and $ b $ are certain functions, and $ |\omega_i| < 1 $, $ i = 1,\dots, n $.
	Here $ \big(W_i(t)\big)_{i=1}^n $ are correlated whereas the standard Brownian motion $ W_0(t) $ is independent of each of the $ \big(W_i(t)\big)_{i=1}^n $.
	The process $ \boldsymbol{W}^*(t) = (\boldsymbol{W}^\top(t), W_0(t))^\top $ is a standard Brownian motion adapted to the natural filtration $ \mathcal{F}(t) = \sigma\big(\boldsymbol{W}^*(s): 0<s<t\big) $ where $ \sigma $ is $ \sigma $-algebra.
	In addition, we denote the constant correlation matrix of $ \boldsymbol{W}(t) $ by $ \big(\rho_{ij}\big)_{n\times n} $.
	More descriptions on these functions will be provided in Assumption \ref{aiai}.  The financial motivation behind this model is driven by the prices of individual stocks in an index (for example, S\&P 500).

	We denote $ \pi_i(t) $ as the discounted amount of wealth invested into the $ i $-th risky asset, $ i = 1, \dots, n $.  Also, we consider \emph{long-only portfolios}. 
	We then denote $ \pi_0(t) $ as the discounted amount of wealth invested into the risk-free asset.
	It is clear that the total discounted wealth $ X(t) $ is calculated by $ X(t) = \pi_0(t) + \sum_{i=1}^{n} \pi_i(t) $.
	If we identify the portfolio by $ n $-tuple $ \big(\pi_i(t)\big)_{i=1}^n $, then we define the discounted wealth process for each tuple by 
	\begin{align*}
		X_i(t) = \frac{\pi_0(t)}{n} +\pi_i(t), \quad i = 1, \dots, n,
	\end{align*}
	such that $ X(t) = \sum_{i=1}^{n}X_i(t) $.
Let $ R $ be the constant non-risky interest rate. 	As we are considering \emph{long positions} in the portfolio, then the discounted amounted $\pi_i$ ($i=1,\dots,n$) invested in each asset can always be equivalently written in terms of the proportions $p_i$ ($i=1,\dots,n$), where $p_i \in (0,1)$ is the proportion invested in the $i$-th stock. Consequently, $ \pi_i(t)=p_i \exp(-tR) S_i(t) $, for all $ i = 1, \dots, n $. We note the geometric Brownian motion type dynamics in \eqref{koekoe} implies that the asset prices are non-negative.
	With initial condition $ S_i(0) > 0 $, we have $ S_i(t) > 0 $. 	Hence $ \pi_i(t) > 0 $, which immediately implies $ X_i(t) $, $ i = 1, \dots, n $, are strictly positive.

Now, we define the dynamics of the discounted wealth processes $ \boldsymbol{X}(t) = (X_1(t), \cdots, X_n(t))^\top $ of the $ n $-tuple by
	\begin{align} \label{Xt}
		d\boldsymbol{X}(t) = \boldsymbol{\sigma}\big(Y(t)\big) \circ \boldsymbol{\pi}\big(t, \boldsymbol{X}^\top(t), Y(t)\big) \circ \Big[ \boldsymbol{\lambda}\big(Y(t)\big)dt + d\boldsymbol{W}(t) \Big],
	\end{align}
	with $ X_i(0)>0 $, $ i = 1, \dots, n $, where $ \boldsymbol{\pi}(\cdot) = (\pi_1(\cdot), \cdots, \pi_n(\cdot))^\top $, and $ \boldsymbol{\lambda}(\cdot) = (\lambda_1(\cdot), \cdots, \lambda_n(\cdot))^\top $ where $ \lambda_i \big(Y(t)\big) = \frac{\mu_i(Y(t)) - R}{ \sigma_i(Y(t)) }$, $ i = 1, \dots, n $, are the Sharpe ratios.
	Here we separate the total wealth $ X(t) $ into $ n $-tuple $ \big(X_i(t)\big)_{i=1}^n $,
	because it's convenient to consider and add the correlations $ (\rho_{ij})_{i,j=1}^n $ occurring among $ \big(W_i(t)\big)_{i=1}^n $ of $ \big(X_i(t)\big)_{i=1}^n $ to the Hamilton-Jacobi-Bellman (HJB) equation \eqref{HJB_0}.

	\begin{assumption} [Model Assumptions]
\label{aiai}
		Let $ C^k(\mathbb{R}) $, $ k \in \mathbb{N} $, be the space of the $ k $-th order continuously differentiable functions. It is assumed that the coefficients in stochastic differential equations \eqref{St} and \eqref{Yt}, as well as the market price of risk $ (\lambda_i)_{i=1}^n $ satisfy:
		\begin{enumerate}
			\item $ \mu_i, \sigma_i \in C(\mathbb{R}) $, $ i = 1, \dots, n $; $ b \in C^1(\mathbb{R}) $; $ a, \lambda_i \in C^2(\mathbb{R}) $, $ i = 1, \dots, n $;
			\item The $ a $ and $ \sigma_i $, $ i = 1, \dots, n $, are strictly positive;
			\item The $ b, b', a, 1/a, a', a'', \lambda_i, \lambda_i' $, and $ \lambda_i'' $, $ i = 1, \dots, n $, are absolutely bounded.
		\end{enumerate}
	\end{assumption}
	
	Similar to Assumption 1 in \cite{Nadtochiy}, the above model assumptions ensure that the system of stochastic differential equations (SDEs) \eqref{St} and \eqref{Yt} has a unique strong solution.
	
	\begin{definition} [Admissible Portfolio] \label{def_1}
		A portfolio $ \boldsymbol{\pi}(t) = \boldsymbol{\pi}\big(t, \boldsymbol{X}^\top(t), Y(t)\big) $  is admissible if
		\begin{enumerate}
			\item it is progressively measurable w.r.t. natural filtration $ \mathcal{F}(t) = \sigma\big(\boldsymbol{W}^*(s): 0<s<t\big) $;
			 \item it is locally square-integrable, i.e. for $ i = 1, \dots, n $, $\displaystyle{\mathbb{E} \bigg( \int_0^T \sigma_i^2(t) \pi_i^2(t) dt \bigg) < \infty}$;
			\item for initial wealth $ x_i \in (0, \infty) $, $ i = 1, \dots, n $, the discounted wealth process \eqref{Xt} is strictly positive, for all $ t \in [0, T] $;
			\item for $ i = 1, \dots, n $, and for $ r >0 $, $\displaystyle{\mathbb{E} \bigg( \int_0^T \frac{\sigma_i^2(t) \pi_i^2(t)}{X_i^{2r}(t)} dt \bigg) < \infty}$ where $ \sigma_i(t) = \sigma_i\big(Y(t)\big) $.
		\end{enumerate}
	\end{definition}

Let $ U = U\big(t, \boldsymbol{X}^\top(t), Y(t)\big) $ be the utility function, and $ J = J\big(t, \boldsymbol{X}^\top(t), Y(t)\big) $ be the value function. Over the set $ \mathcal{A} $ of all admissible portfolios, we define the value function by
	\begin{align} \label{J}
		J(t,\boldsymbol{x},y) = \esssup_{\boldsymbol{\pi}\in \mathcal{A}} \mathbb{E}\Big(U\big(T, \boldsymbol{X}^\top(T)\big) \Big| \boldsymbol{X}^\top(t) = \boldsymbol{x}, Y(t) = y \Big),
	\end{align}
where $ \boldsymbol{x} = (x_1,\dots,x_n) $, $ T $ is terminal time, and $ \mathbb{E} $ is the expection function.

	For convenience, we denote $ b = b(y) $, $ a = a(y) $, and $ W_i = W_i(t) $, $ \sigma_i = \sigma_i(y) $, $ \pi_i = \pi_i(t, \boldsymbol{x}, y) $, $ \lambda_i = \lambda_i(y) $, $ \forall i $. From the system of SDEs \eqref{Yt} and \eqref{Xt}, we have
	\begin{align*}
		\begin{bmatrix} dt \ \ \\ dx_1 \\ \vdots \\ dx_n \\dy \ \ \end{bmatrix}  = 
		\begin{bmatrix} 1 \\ \sigma_1 \pi_1 \lambda_1 \\ \vdots \\ \sigma_n \pi_n \lambda_n \\ b \end{bmatrix}dt 
		+ \begin{bmatrix} 0 \\ \sigma_1 \pi_1 \\ \vdots \\ 0 \\ a \omega_1 \end{bmatrix} dW_1
		+ \cdots
		+ \begin{bmatrix} 0 \\ 0 \\ \vdots \\ \sigma_n \pi_n  \\ a \omega_n \end{bmatrix} dW_n
		+ \begin{bmatrix} 0 \\ 0 \\ \vdots \\ 0 \\ a \big[1-\sum_{i=1}^n \omega_i^2\big]^{1/2} \end{bmatrix}dW_0.
	\end{align*}
	As the dependence merely occurs among $ \big( W_i(t) \big)_{i=1}^n $,
	the value function $ J $ is \emph{formally} a solution to the HJB equation given by
	\begin{align} \label{HJB_0}
	 	U_t & + \max_{\boldsymbol{\pi}} \Bigg(
	 	\sum_{i=1}^{n}\sigma_i  \pi_i \lambda_i U_{x_i} + \frac{1}{2} \sum_{i,j=1}^{n} \rho_{ij} \sigma_i \pi_i  \sigma_j \pi_j U_{x_ix_j} + a\sum_{i=1}^{n} \sigma_i \pi_i \omega_i  U_{x_iy}
	 	\Bigg)
	 	+ bU_y + \frac{1}{2}a^2U_{yy} = 0
	\end{align}
	with terminal condition $ U(T,\boldsymbol{x},y) = U(T,\boldsymbol{x}) = U_T(\boldsymbol{x}) $.  We assume that $ U_T = U_T(\boldsymbol{x}) $ has the properties as below. 
	
In some previous works (see \cite{Kumar, LinSG, Nadtochiy}) similar model assumptions are used to construct approximate solutions of the corresponding HJB equation. In particular, the approximate solutions for the papers \cite{Kumar} and \cite{LinSG} already incorporate the fact that corresponding HJB equation admits a unique classical solution. For \eqref{HJB_0}, since we implement similar model and assumption, it admits unique classical solution. Consequently, the probabilistic setup of this paper is same as \cite{Kumar} and \cite{LinSG}.

	\begin{definition}
		For the rest of the paper we denote $ \| \boldsymbol{x} \| := \sum_{i=1}^{n} |x_i| = \sum_{i=1}^{n} x_i $, for $x_i \geq 0$, $i=1, \dots,n$. 
	\end{definition}
	
	\begin{assumption} [Utility Assumptions] \label{Assume_UT}
	 	It is assumed that $ U_T(\boldsymbol{x}) $ is strictly increasing along $ \| \boldsymbol{x} \| $ and concave in $ C^5\big((\mathbb{R}^+)^n\big) $, and 
	 	$ U_T(\boldsymbol{x}) $ behaves as either of the following cases, asymptotically as $ \| \boldsymbol{x} \| $ approaches to both $ 0 $ and $ \infty $:
	 	\begin{enumerate}
	 		\item Logarithmic function: $ U_T(\boldsymbol{x}) = \ln \| \boldsymbol{x} \| $.
	 		\item Mixture of power functions: $ U_T(\boldsymbol{x}) = \frac{c_1}{1-\alpha}\| \boldsymbol{x} \|^{1-\alpha} + \frac{c_2}{1-\beta} \| \boldsymbol{x} \|^{1-\beta} $, where $ \alpha, \beta > 0 $, $ \alpha, \beta \neq 1 $, and  $ c_1>0,  c_2=0 $, or $ c_1=0,  c_2>0 $, or $ c_1, c_2>0 $.
	 	\end{enumerate}
	 \end{assumption}
	
	 If the system of SDEs \eqref{Yt} and \eqref{Xt} is Markovian, then the expression being maximized in HJB equation \eqref{HJB_0} achieves its maximum at the optimal portfolios $ \big(\widehat{\pi}_i\big)_{i=1}^n $. 
	 
We first observe that the Hamiltonian in \eqref{HJB_0} is quadratic in $ \boldsymbol{\pi} $ and thus its optimum $ \widehat{\boldsymbol{\pi}} $ can be found with an explicit expression. For this, we rewrite this expression as
	 \begin{align*}
	 	\boldsymbol{\phi}^\top \boldsymbol{\pi} + \frac{1}{2} \boldsymbol{\pi}^\top \boldsymbol{H}\boldsymbol{\pi} + a \boldsymbol{\psi}^\top \boldsymbol{\pi}
	 \end{align*}
	 where $ \boldsymbol{\phi} = (\sigma_1 \lambda_1 U_{x_1}, \cdots, \sigma_n \lambda_n U_{x_n})^\top $, $ \boldsymbol{\psi} = (\sigma_1 \omega_1  U_{x_1y}, \cdots, \sigma_n \omega_n  U_{x_ny})^\top $, and $ \boldsymbol{H} = (\rho_{ij} \sigma_i  \sigma_j  U_{x_ix_j})_{n \times n} $.
	 By the first-order condition, we obtain
	 \begin{align*}
	 	\widehat{\boldsymbol{\pi}} = \boldsymbol{H}^{-1} \big[-\boldsymbol{\phi}-a\boldsymbol{\psi}\big].
	 \end{align*}
	 In the case of double risky assets portfolio, i.e. $ n=2 $, 
	 \begin{align*}
	 	\widehat{\pi}_1 & = \frac{-\big[\lambda_1 U_{x_1}+a\omega_1 U_{x_1y}\big]U_{x_2x_2} + \rho_{12}U_{x_1x_2}\big[\lambda_2U_{x_2} + a \omega_2 U_{x_2y}\big]}
	 	{\sigma_1\big[U_{x_1x_1}U_{x_2x_2} - \rho_{12}^2 U_{x_1x_2}^2\big]},\\
	 	\widehat{\pi}_2 & = \frac{-\big[\lambda_2 U_{x_2}+a\omega_2 U_{x_2y}\big]U_{x_1x_1} + \rho_{12}U_{x_1x_2}\big[\lambda_1U_{x_1} + a \omega_1 U_{x_1y}\big]}
	 	{\sigma_2\big[U_{x_1x_1}U_{x_2x_2} - \rho_{12}^2 U_{x_1x_2}^2\big]}.	
	 \end{align*}

However, 	in the general case, i.e. for arbitrary $ n $, it is not possible to find a straighfoward explicit expression of $ \boldsymbol{H}^{-1} $. Consequently, for the general case, we  solve for the individual optimal portfolio $ \widehat{\pi}_i $.	 	 
	
	By the first-order condition, we obtain
	\begin{align} \label{pi}
		\widehat{\pi}_i =  
		\sum_{\substack{j=1 \\ j \neq i}}^n \Bigg[ - \frac{\sigma_j \rho_{ij} U_{x_ix_j} \widehat{\pi}_j}{2 \sigma_i U_{x_ix_i}} \Bigg]
		+ \frac{-\lambda_iU_{x_i}-a \omega_i U_{x_iy}}{\sigma_i U_{x_ix_i}}
		, \quad i = 1, \dots, n.
	\end{align}
	Note that equation \eqref{pi} can be solved using the Jacobi iterative method or the Gauss-Siedel iterative method. It is clear that \eqref{pi} can be written as a liner system, where the $ i $-th equation is given by: $ \textbf{U} \boldsymbol{\widehat{\pi}} = \boldsymbol{v} $ for $ \widehat{\pi}_i $, where $ \textbf{U} = (u_{ij})_{n\times n} $, $ \boldsymbol{\widehat{\pi}} = (\widehat{\pi}_1, \dots, \widehat{\pi}_n)^\top $, $ \boldsymbol{v} = (v_1, \dots, v_n)^\top $,
	\begin{align*}
		u_{ii} = U_{x_ix_i}, \
		u_{ij} = \frac{\sigma_j \rho_{ij} U_{x_ix_j}}{2\sigma_i}, i \neq j, 
		\text{ and } v_i = \frac{-\lambda_iU_{x_i}-a \omega_i U_{x_iy}}{\sigma_i}.
	\end{align*}
	With $ (\| \boldsymbol{x} \|)_{x_i} = 1 $, $ \forall i $, we observe
	\begin{align*}
		(U_T)_{x_i} = (U_T)_{x_j} \text{ and } (U_T)_{x_ix_i} = (U_T)_{x_ix_j}, i,j = 1, \dots, n, i \neq j.
	\end{align*}
	Since the primary objective of this paper is to approximate the value function by an expansion of terminal utility function, we denote $ U = U\big(t,U_T,y\big) $ under terminal condition. Then
	\begin{align*}
		U_{x_ix_i} = U_{x_ix_j} = \frac{\partial^2 U}{\partial (U_T)^2}\big[(U_T)_{x_i}\big]^2 + \frac{\partial U}{\partial U_T}(U_T)_{x_ix_i}, \quad i,j = 1, \dots, n, \quad i \neq j. 
	\end{align*}
	Hence
	\begin{align} \label{iiij}
		u_{ij} = \frac{\sigma_j \rho_{ij}}{2\sigma_i}u_{ii} = \frac{\sigma_j \rho_{ij}}{2\sigma_i}u_{jj}, \quad i,j = 1, \dots, n, \quad i \neq j.
	\end{align}
	Using equations \eqref{iiij}, we assume the volatility functions have the following property such that $ \textbf{U} $ is a strictly diagonally dominant matrix.
	\begin{assumption} \label{Conv}
		Let $ \rho_{ij} $ be the correlation between standard Brownian motions $ W_i $ and $ W_j $, $ i,j = 1, \dots, n $, $ i \neq j $. Let $ \sigma_i(y) $, $ i = 1, \dots, n $, be the volatility functions where $ y = Y(t) \in \mathbb{R} $. We assume that for each $ i = 1, \dots, n $, the following inequality holds:
		\begin{align*}
			\sum_{j=1}^n |\rho_{ij}| \sigma_j(y) < 3\sigma_i(y).
		\end{align*}
	\end{assumption}

	\noindent Although not all portfolios satisfy the above assumption, there are many cases for which an investor can build a portfolio satisfying it. Based on one-year data and half-year data, several such examples are provided in Appendix \ref{apendixC}.
	
	We denote $ \textbf{D} = diag(u_{11},\dots, u_{nn}) $,
	\begin{align*}
		\textbf{L} = \begin{bmatrix} 
			0 & \cdots & \cdots & 0 \\ 
			-u_{21} & \ddots &  & \vdots \\ 
			\vdots & \ddots & \ddots & \vdots \\ 
			-u_{n1} & \cdots & -u_{n,n-1} & 0\end{bmatrix}, 
		\text{ and } 
		\textbf{V} = \begin{bmatrix} 
			0 & -u_{12} & \cdots & -u_{1n} \\ 
			\vdots & \ddots & \ddots & \vdots \\ 
			\vdots &  & \ddots & -u_{n-1,n} \\ 
			0 & \cdots & \cdots & 0\end{bmatrix}.
	\end{align*}
	If we choose the Gauss-Seidel method, then it is necessary to have an additional approximation to expand $ (\textbf{D} - \textbf{L})^{-1} $. 
	So we choose the Jacobi method, and define
	\begin{align*}
		\textbf{T} = \textbf{D}^{-1}(\textbf{L}+\textbf{V}) =
		\begin{bmatrix} 
			0 & -\frac{\sigma_2 \rho_{12}}{2\sigma_1} & \cdots & -\frac{\sigma_n\rho_{1n}}{2\sigma_1} \\ 
			-\frac{\sigma_1\rho_{21}}{2\sigma_2} & 0 & \ddots & \vdots \\ 
			\vdots & \ddots & \ddots & -\frac{\sigma_n\rho_{n-1,n}}{2\sigma_{n-1}}\\ 
			-\frac{\sigma_1\rho_{n1}}{2\sigma_n} & \cdots & -\frac{\sigma_{n-1}\rho_{n,n-1}}{2\sigma_n} & 0
		\end{bmatrix}.
	\end{align*}

	In the following, we summarize the Lemma 7.18, Theorems 7.19 and 7.21 from \cite{Burden}. These theorems will be utilized for the rest of this section. 
	
	\begin{theorem} \label{7.21}
		If $ \textbf{A} $ is strictly diagonally dominant, then for any choice of $ \textbf{x}^{(0)} $, both the Jacobi and Gauss-Seidel methods give the  sequence of vectors $ \{ \textbf{x}^{(k)} \}^\infty_{k=0} $ that converge to the unique solution of linear system $  \textbf{A}\textbf{x} = \textbf{b} $.
	\end{theorem}
	\begin{theorem} \label{7.19}
		For any vector $ \textbf{x}^{(0)} \in \mathbb{R}^n $, the sequence of vectors $ \{\textbf{x}^{(k)}\}^\infty_{k=0} $ deﬁned by
		\begin{align*}
			\textbf{x}^{(k)} = \textbf{T} \textbf{x}^{(k-1)} + \textbf{c}, \text{ for each } k \geq 1,
		\end{align*}
		converges to the unique solution of $ \textbf{x} = \textbf{T} \textbf{x} + \textbf{c} $ if and only if $ \rho(\textbf{T}) < 1 $.
	\end{theorem}
	\begin{theorem} \label{7.18}
		If the spectral radius satisﬁes $ \rho(\textbf{T}) < 1 $, then $ (\textbf{I} - \textbf{T})^{-1} $ exists, and
		\begin{align*}
			(\textbf{I} - \textbf{T})^{-1} = \textbf{I} + \textbf{T} + \textbf{T}^2 + \cdots = \sum_{j=0}^{\infty} \textbf{T}^j.
		\end{align*}
	\end{theorem}
	Under Assumption \ref{Conv}, $ \textbf{U} $ is strictly diagonally dominant. 
	Then, Theorem \ref{7.21} implies the sequence $ \big( \boldsymbol{\widehat{\pi}}^{(l)} \big)_{l=0}^\infty $ given by Jacobi method converges to the unique solution of $ \textbf{U} \boldsymbol{\widehat{\pi}} = \boldsymbol{v} $, which is equivalent to the linear system $ \boldsymbol{\widehat{\pi}} = \textbf{T} \boldsymbol{\widehat{\pi}} + \textbf{D}^{-1} \boldsymbol{v} $. 
	Next, Theorem \ref{7.19} implies the spectral radius satisfies $ \rho(\textbf{T}) <1 $. 
	After that, Theorem \ref{7.18} implies that $ (\textbf{I} - \textbf{T})^{-1} $ exists and $ (\textbf{I} - \textbf{T})^{-1} = \textbf{I} + \textbf{T}  +O(\textbf{T}^2) $, where $ \textbf{I} $ is identity matrix. Hence
	\begin{align*}
		\widehat{\boldsymbol{\pi}} 
		& = (\textbf{I} - \textbf{T})^{-1} \textbf{D}^{-1} \boldsymbol{v} \\
		& = (\textbf{I} + \textbf{T}) \textbf{D}^{-1} \boldsymbol{v} \\
		& = 
		\begin{bmatrix} 
			1 & -\frac{\sigma_2 \rho_{12}}{2\sigma_1} & \cdots & -\frac{\sigma_n\rho_{1n}}{2\sigma_1} \\ 
			-\frac{\sigma_1\rho_{21}}{2\sigma_2} & 1 & \ddots & \vdots \\ 
			\vdots & \ddots & \ddots & -\frac{\sigma_n\rho_{n-1,n}}{2\sigma_{n-1}}\\ 
			-\frac{\sigma_1\rho_{n1}}{2\sigma_n} & \cdots & -\frac{\sigma_{n-1}\rho_{n,n-1}}{2\sigma_n} & 1
		\end{bmatrix}
		\begin{bmatrix} 
			\frac{v_1}{u_{11}} \\ 
			\vdots \\ 
			\vdots \\ 
			\frac{v_n}{u_{nn}}
		\end{bmatrix}.
	\end{align*}
	It follows that, for $ i = 1, \dots, n $, the optimal portfolios are:
	\begin{align} \label{pi_hat}
		\widehat{\pi}_i 
		& = \frac{v_i}{u_{ii}} 	
		- \frac{1}{2\sigma_i u_{ii}} \sum_{\substack{k=1 \\ k \neq i}}^n \sigma_k \rho_{ik} v_k
		\nonumber  \\
		& = \frac{-\lambda_iU_{x_i}-a \omega_i U_{x_iy}}{\sigma_i U_{x_ix_i}}
		+ \frac{1}{2\sigma_i U_{x_ix_i}} \sum_{\substack{k=1 \\ k \neq i}}^n 
		\rho_{ik} \Big[\lambda_kU_{x_k}+a \omega_k U_{x_ky}\Big]
		\nonumber \\
		& = \frac{1}{2\sigma_i U_{x_ix_i}} 
		\Bigg[
		-3\Big[\lambda_iU_{x_i}+a \omega_i U_{x_iy}\Big]
		+  \sum_{k=1}^n 
		\rho_{ik} \Big[\lambda_kU_{x_k}+a \omega_k U_{x_ky}\Big]
		\Bigg].
	\end{align}
	Finally, after substituting \eqref{pi_hat} into \eqref{HJB_0}, we obtain the following HJB equation
	\begin{align} \label{HJB}
		\left\{
			\begin{array}{ll}
				U_t + \mathcal{H}(U)  =0, & \text{for } (t,\boldsymbol{x},y) \in [0,T]\times (\mathbb{R}^+)^n \times \mathbb{R},\\
				U(T,\boldsymbol{x},y) = U_T(\boldsymbol{x}), & \text{for } (\boldsymbol{x},y) \in (\mathbb{R}^+)^n \times \mathbb{R}, \\
			\end{array}
		\right.
	\end{align}
	where
	\begin{align*}
		\mathcal{H}(U) =
		\sum_{i=1}^{n}  \sigma_i\widehat{\pi}_i \lambda_i U_{x_i} + \frac{1}{2} \sum_{i,j=1}^{n} \rho_{ij} \sigma_i \sigma_j\widehat{\pi}_i \widehat{\pi}_j U_{x_ix_j} + a\sum_{i=1}^{n} \sigma_i \widehat{\pi}_i \omega_i  U_{x_iy}
		+ bU_y + \frac{1}{2}a^2U_{yy}.
	\end{align*}

	\section{Main result}
\label{sec3}

	In this section, we construct classical super-solution and sub-solution to HJB equation \eqref{HJB} using the second order expansion of utility function in powers of the time to horizon $ T - t $.
	We approximate the value function defined by \eqref{J} using the first order terms of expansion in power of the time to horizon $ T - t $. 
	We then control the error of this approximation by the second order terms of expansion in powers of the square time to horizon $ (T - t)^2 $.
	We also prove that value function lies between constructed super-solution and sub-solution using martingale inequalities. 

	\begin{theorem} \label{thm}
		Let $ \widehat{U}(t,\boldsymbol{x},y) $ be the first-order approximate solution of the HJB equation \eqref{HJB}. Then
		\begin{align} \label{U_hat}
			\widehat{U}(t,\boldsymbol{x},y) = U_T(\boldsymbol{x}) + (T-t)U^{(1)}(\boldsymbol{x},y),
		\end{align}
		where
		\begin{align} \label{U^1}
			U^{(1)}(\boldsymbol{x},y) = \
			& \frac{1}{2}\sum_{i=1}^{n} \lambda_i(y)
			\Big[ \boldsymbol{\rho}_{i} \boldsymbol{\lambda}(y) -3\lambda_i(y) \Big] 
			\frac{\big[\big(U_T(\boldsymbol{x})\big)_{x_i}\big]^2}{\big(U_T(\boldsymbol{x})\big)_{x_ix_i}} \nonumber \\
			& + \frac{1}{8} \sum_{i,j=1}^{n} \rho_{ij}
			\Big[ \boldsymbol{\rho}_{i} \boldsymbol{\lambda}(y) -3\lambda_i(y) \Big]
			\Big[ \boldsymbol{\rho}_{j} \boldsymbol{\lambda}(y) -3\lambda_j(y) \Big]
			\frac{\big[\big(U_T(\boldsymbol{x})\big)_{x_i}\big]^2}{\big(U_T(\boldsymbol{x})\big)_{x_ix_i}}
		\end{align}
		where $ \boldsymbol{\rho}_{i} = (\rho_{i1}, \cdots, \rho_{in}) $ is the row vector of constant correlation matrix $ \big(\rho_{ij}\big)_{n\times n} $ of $ \boldsymbol{W}(t) $.

		Moreover, there exists constants $ c > 0 $ and $ 0 < \varepsilon < \min\{1, T\} $ such that
		\begin{align} \label{ineq}
			\big| J(t,\boldsymbol{x},y) - \widehat{U}(t,\boldsymbol{x},y) \big| \leq c (T-t)^2 f(\boldsymbol{x}) \quad \text{for }  (t,\boldsymbol{x},y) \in (T-\varepsilon,T)\times (\mathbb{R}^+)^n \times \mathbb{R},
		\end{align}
		where $ J(t,\boldsymbol{x},y) $ is value function defined by \eqref{J}; $ f(\boldsymbol{x}) = 1 $ under Case 1 of Assumption \ref{Assume_UT}, and $ f(\boldsymbol{x}) = \|\boldsymbol{x}\|^{1-\alpha} + \|\boldsymbol{x}\|^{1-\beta} $ under Case 2 of Assumption \ref{Assume_UT}; the constants $ c $ and $ \varepsilon$ are independent of $ t $, $ \boldsymbol{x} $ and $ y $.

	\end{theorem}

We provide a sketch of proof of Theorem \ref{thm}. The detailed proof is provided in Appendix \ref{apendixA1}. 
\begin{proof}[Sketch of proof]
We prove this theorem in two steps. In the first step, we construct a super-solution and a sub-solutions to HJB equation \eqref{HJB}. For this, at first \eqref{U^1} is derived. After that, an expression is computed for $ U^{(2)} $, 
	which is the coefficient of the second order term of the power expansion of utility function in the time to horizon $ T - t $. It is shown that all the terms of $ U^{(2)} $ are of the same order as $ f(\boldsymbol{x}) $. The details of the notation is provided in the detailed proof in Appendix \ref{apendixA1}. 
After this, we define the super-solution and the sub-solution. Based on those, we show that the coefficients of $ (T-t) $-terms are of the order of $ \tilde{f}(\boldsymbol{x}) $, for an appropriate $ \tilde{f}(\boldsymbol{x}) $. It is also shown that the coefficients are bounded. After that, combining all these results, it is proved that the super-solution and the sub-solution are the classical super-solution/sub-solution of HJB equation \eqref{HJB}.

For the second step,  we prove $ \underline{U}(t,\boldsymbol{x},y) \leq J(t,\boldsymbol{x},y) \leq \overline{U}(t,\boldsymbol{x},y) $ (for the notations, we once again refer to the detailed proof in Appendix \ref{apendixA1}). To do this, at first a sequence of measurable and integrable real functions $ \big\{\overline{U}\big(t_n, \boldsymbol{X}(t_n), Y(t_n)\big)\big\}_{n=1}^\infty $, are constructed. Then a related inequality is proved, and the assumptions of the dominated convergence theorem are verified. After that the dominated convergence theorem is used to prove an asymptotic property to the constructed sequence of functions. Combining all these results, it is proved that $ J(t,\boldsymbol{x},y) \leq \overline{U}(t,\boldsymbol{x},y) $. Similarly, we prove $ \underline{U}(t,\boldsymbol{x},y) \leq J(t,\boldsymbol{x},y) $. From all these results, \eqref{ineq} is established.
\end{proof}

We conclude this section with a special case. 
If $ n = 1 $, then the first-order approximate solution of HJB equation \eqref{HJB} in Theorem \eqref{thm} is reduced as 
\begin{align*}
	\widehat{U}(t,x,y) = U_T(x) - (T-t)\frac{\lambda^2(y)\big(U'_T(x)\big)^2}{2U''_T(x)}.
\end{align*}
We also control the error of this reduced approximation by the following: there exists constants $ c > 0 $ and $ 0 < \varepsilon < \min\{1, T\} $ such that
\begin{align*}
	\big| J(t,x,y) - \widehat{U}(t,x,y) \big| \leq c (T-t)^2 f(x) \quad \text{for }  (t,x,y) \in (T-\varepsilon,T)\times \mathbb{R}^+ \times \mathbb{R},
\end{align*}
where $ J(t,x,y) $ is value function defined by \eqref{J}; $ f(x) = 1 $ under Case 1 of Assumption \ref{Assume_UT}, and $ f(x) = x^{1-\alpha} + x^{1-\beta} $ under Case 2 of Assumption \ref{Assume_UT}; the constants $ c $ and $ \varepsilon$ are independent of $ t $, $ x $ and $ y $.

We note the above reduction is the same as the primary result in \cite{Kumar}. Hence our $ n $-dimensional model includes all of the advantages of its corresponding $ 1 $-dimensional model in \cite{Kumar}.

\section{Approximating portfolio}

\subsection{Generating the approximating portfolio}
\label{sec4}

In this subsection, we generate a close-to-optimal portfolio near the time to horizon $ (T-t) $ by substituting the first-order approximation of utility function \eqref{U_hat} into the optimal portfolio \eqref{pi_hat}. 
In the following lemma, we also prove that the generated portfolio yields an expected utility function close to the maximum expected utility, and  the error is controlled by the square of the time to the horizon $ (T-t)^2 $.
In the financial context, this will provide a practical way to obtain an optimal portfolio, with control on the errors. 

	\begin{theorem}
		Let $ \tau \in [t, T] $. Denote $ \boldsymbol{\hat{X}}(\tau) = \big(\widehat{X}_1(\tau),\dots,\widehat{X}_n(\tau)\big) $. For $ i = 1,\dots,n $, the wealth $ \widehat{X}_i(\tau) $ is 
		\begin{align*}
			d\widehat{X}_i(\tau) = \sigma_i\big(Y(\tau)\big) \widetilde{\pi}_i\big(\tau,\boldsymbol{\hat{X}}(\tau),Y(\tau)\big) \Big[ \lambda_i\big(Y(\tau)\big)dt + dW_i(\tau) \Big],
		\end{align*}
		where the optimal portfolio $ \widetilde{\pi}_i\big(\tau,\boldsymbol{\hat{X}}(\tau),Y(\tau)\big) $ is given by
		\begin{align*}
			& \widetilde{\pi}_i\big(\tau, \boldsymbol{\hat{X}}(\tau), Y(\tau)\big) = 
			\frac{1}{2\sigma_i\big(Y(\tau)\big) \widehat{U}_{\widehat{X}_i\widehat{X}_i}\big(\tau,\boldsymbol{\hat{X}}(\tau),Y(\tau)\big)} \cdot\\
			& \cdot
			\Bigg[
			-3\Big[\lambda_i\big(Y(\tau)\big) \widehat{U}_{\widehat{X}_i}\big(\tau,\boldsymbol{\hat{X}}(\tau),Y(\tau)\big) + \omega_i a\big(Y(\tau)\big)  \widehat{U}_{\widehat{X}_iY}\big(\tau,\boldsymbol{\hat{X}}(\tau),Y(\tau)\big)\Big] \\
			& \quad \
			+  \sum_{k=1}^n 
			\rho_{ik} \Big[\lambda_k\big(Y(\tau)\big) \widehat{U}_{\widehat{X}_k}\big(\tau,\boldsymbol{\hat{X}}(\tau),Y(\tau)\big) + \omega_k a(Y(\tau)\big)  \widehat{U}_{\widehat{X}_kY}\big(t,\boldsymbol{\hat{X}}(\tau),Y(\tau)\big)\Big]
			\Bigg]
		\end{align*}
		where the utility function $ \widehat{U} $ is given by approximation \eqref{U_hat}. 

		Moreover, there exists constant $ c > 0 $ and $ 0 < \varepsilon < \min\{1, T\} $ such that
		\begin{multline*}
			\bigg| J(t,\boldsymbol{x},y) - \mathbb{E}\Big( U_T\big(\boldsymbol{\hat{X}}(T)\big) \Big| \boldsymbol{\hat{X}}(t) = \boldsymbol{x}, Y(t) = y \Big) \bigg| \leq c (T-t)^2f(\boldsymbol{x}) \\ 
			\text{for }  (t,\boldsymbol{x},y) \in (T-\varepsilon,T)\times (\mathbb{R}^+)^n \times \mathbb{R},
		\end{multline*}
		where $ J(t,\boldsymbol{x},y) $ is value function defined by \eqref{J}; $ f(\boldsymbol{x}) = 1 $ under Case 1 of Assumption \ref{Assume_UT}, and $ f(\boldsymbol{x}) = \|\boldsymbol{x}\|^{1-\alpha} + \|\boldsymbol{x}\|^{1-\beta} $ under Case 2 of Assumption \ref{Assume_UT}; the constants $ c $ and $\varepsilon$ are independent of $ t $, $ \boldsymbol{x} $ and $ y $.
	\end{theorem}
	\begin{proof}
		We denote $ \lambda_i = \lambda_i\big(Y(\tau)\big) $, $ \sigma_i = \sigma_i\big(Y(\tau)\big) $, $ b = b\big(Y(\tau)\big) $, $ a = a\big(Y(\tau)\big) $, $ \widetilde{\pi}_i = \widetilde{\pi}_i\big(\tau, \boldsymbol{\hat{X}}(\tau), Y(\tau)\big) $ and $ \widehat{U} = \widehat{U}\big(\tau, \boldsymbol{\hat{X}}(\tau), Y(\tau)\big) $. We let $ \{t_n \}_{n=1}^\infty \subset [t, T] $ be a sequence of stopping times such that $ t_n \leq t_{n+1} $ and $ t_n \rightarrow T $. 
		
		We apply the multidimensional It\^{o}'s formula to $ \widehat{U} $ and obtain
		\begin{align}
			& \widehat{U}\big(t_n, \boldsymbol{\hat{X}}(t_n), Y(t_n)\big) - \widehat{U}(t,\boldsymbol{x},y)  \nonumber\\
			& = \int_t^{t_n} \bigg[ \widehat{U}_{\tau} + \sum_{i=1}^{n}\sigma_i \widetilde{\pi}_i \lambda_i \widehat{U}_{\widehat{X}_i} + b\widehat{U}_Y
			+ \frac{1}{2} \sum_{i,j=1}^{n} \rho_{ij} \sigma_i \widetilde{\pi}_i  \sigma_j \widetilde{\pi}_j \widehat{U}_{\widehat{X}_i\widehat{X}_j} + a\sum_{i=1}^{n} \sigma_i \widetilde{\pi}_i \omega_i  \widehat{U}_{\widehat{X}_iY} +\frac{1}{2}a^2\widehat{U}_{YY} \bigg] ds  \label{HJB_5}\\
			& \quad
			+ \sum_{i=1}^{n} \int_t^{t_n} \bigg[ \sigma_i \widetilde{\pi}_i \widehat{U}_{\widehat{X}_i} + a \omega_i \widehat{U}_Y \bigg] dW_i
			+ \int_t^{t_n} \bigg[ a \Big[1-\sum_{i=1}^n \omega_i^2\Big]^{1/2} \widehat{U}_Y \bigg] dW_0. \label{mart_5}
		\end{align}
		where the stochastic integrals \eqref{mart_5} are martingales. We note that the integrand of \eqref{HJB_5} is exactly the $ \widehat{U}_\tau + \mathcal{H}(\widehat{U}) $. 
		As $ U^{(1)} $ is solved by ignoring the $ (T-t) $-terms, $ \big| \widehat{U}_\tau + \mathcal{H}(\widehat{U}) \big| = O(T-\tau)O\big(f(\boldsymbol{\hat{X}}(\tau))\big) $ under Assumption \ref{Assume_UT}.
		By the martingale property, we have
		\begin{multline*}
			\mathbb{E}\Big( \widehat{U}\big(t_n, \boldsymbol{\hat{X}}(t_n), Y(t_n)\big) \Big| \boldsymbol{\hat{X}}(t) = \boldsymbol{x}, Y(t) = y \Big) - \widehat{U}(t,\boldsymbol{x},y) \\
			= \int_t^{t_n} \mathbb{E} \Big( O(T-\tau)O\big(f\big(\boldsymbol{\hat{X}}(\tau)\big)\big) \Big| \boldsymbol{\hat{X}}(t) = x, Y(t) = y \Big) d\tau.
		\end{multline*}
		Referring to Lemma \ref{A.1}, $ f(\boldsymbol{\hat{X}}) $ is also uniform bounded. Then we have
		\begin{align} \label{4.3}
			\bigg| \mathbb{E}\Big( \widehat{U}\big(t_n, \boldsymbol{\hat{X}}(t_n), Y(t_n)\big) \Big| \boldsymbol{\hat{X}}(t) = \boldsymbol{x}, Y(t) = y \Big) - \widehat{U}(t,\boldsymbol{x},y) \bigg|  \leq c_1 (T-t)^2 f(\boldsymbol{x}),
		\end{align}
		with some constant $ c_1>0 $.
		
		We apply equation \eqref{U_hat} and triangle inequality to $ \widehat{U}\big(t_n, \boldsymbol{\hat{X}}(t_n), Y(t_n)\big) $ and obtain
		\begin{align*}
			\Big| \widehat{U}\big(t_n, \boldsymbol{\hat{X}}(t_n), Y(t_n)\big) \Big|
			& = \Big| U_T\big(\boldsymbol{\hat{X}}(t_n)\big) + [T-t_n]U^{(1)}\big(\boldsymbol{\hat{X}}(t_n),Y(t_n)\big) \Big| \\
			& \leq \Big| U_T\big(\boldsymbol{\hat{X}}(t_n)\big) \Big| + \Big| T U^{(1)}\big(\boldsymbol{\hat{X}}(t_n),Y(t_n)\big) \Big| \\
			& \leq c_2 g\big(\boldsymbol{\hat{X}}(t_n)\big),
		\end{align*}
		with some constant $ c_2 $. 
		Referring to the proof of Theorem \ref{thm} (Appendix \ref{apendixA1}), we have $ U^{(1)} \sim f\big(\boldsymbol{\hat{X}}(t_n)\big) $ under Case 1 of Assumption \ref{Assume_UT}.
		Then $ g\big(\boldsymbol{\hat{X}}(t_n)\big) = U_T\big(\boldsymbol{\hat{X}}(t_n)\big) + f\big(\boldsymbol{\hat{X}}(t_n)\big) = \ln\|\boldsymbol{\hat{X}}(t_n)\| +1 $ under this case.
		Also, we have $ U_T\big(\boldsymbol{\hat{X}}(t_n)\big) \sim f\big(\boldsymbol{\hat{X}}(t_n)\big) $ under Case 2 of Assumption \ref{Assume_UT}. 
		Thus $ g\big(\boldsymbol{\hat{X}}(t_n)\big) = f\big(\boldsymbol{\hat{X}}(t_n)\big) = \|\boldsymbol{\hat{X}}(t_n)\|^{1-\alpha} + \|\boldsymbol{\hat{X}}(t_n)\|^{1-\beta} $ under this case. 
		
		With substituting $ \boldsymbol{\hat{X}}(t_n) $ into $ \boldsymbol{X}(t_n) $, the Lemma \ref{A.1} in the Appendix \ref{apendixA} also proves that $ \big\{ g\big(\boldsymbol{\hat{X}}(t_n)\big) \big\}_{n=1}^\infty $ is uniformly bounded by an integrable random variable.  We already have 
		$$ \Big| \widehat{U}\big(t_n, \boldsymbol{\hat{X}}(t_n), Y(t_n)\big) \Big| \leq c_2 g\big(\boldsymbol{\hat{X}}(t_n)\big),$$ with some constant $ c_2 $. We also observe that
		\begin{align*}
			t_n \rightarrow T\ \Rightarrow\ \widehat{U}\big(t_n, \boldsymbol{\hat{X}}(t_n), Y(t_n)\big) \rightarrow \widehat{U}\big(T, \boldsymbol{\hat{X}}(T), Y(T)\big) = U_T\big(\boldsymbol{\hat{X}}(T)\big).
		\end{align*}
		By the dominated convergence theorem, we have
		\begin{align} \label{lim2}
			\lim_{n \rightarrow \infty} \mathbb{E}\Big( \widehat{U}\big(t_n, \boldsymbol{\hat{X}}(t_n), Y(t_n)\big) \Big| \boldsymbol{\hat{X}}(t) = \boldsymbol{x}, Y(t) = y \Big) = \mathbb{E}\Big( U_T\big(\boldsymbol{\hat{X}}(T)\big) \Big| \boldsymbol{\hat{X}}(t) = \boldsymbol{x}, Y(t) = y \Big).
		\end{align}
		Combing \eqref{4.3} with \eqref{lim2}, we have
		\begin{align} \label{4.5}
			\bigg| \mathbb{E}\Big( U_T\big(\boldsymbol{\hat{X}}(T)\big) \Big| \boldsymbol{\hat{X}}(t) = \boldsymbol{x}, Y(t) = y \Big) - \widehat{U}(t,\boldsymbol{x},y) \bigg|  \leq c_1 (T-t)^2 f(\boldsymbol{x}).
		\end{align}
		By \eqref{ineq}, \eqref{4.5} and triangle inequality, there exists constant $ c > 0 $ and $ 0 < \varepsilon < \min\{1, T\} $ such that
		\begin{align*}
			& \bigg| J(t,\boldsymbol{x},y) - \mathbb{E}\Big( U_T\big(\boldsymbol{\hat{X}}(T)\big) \Big| \boldsymbol{\hat{X}}(t) = \boldsymbol{x}, Y(t) = y \Big)  \bigg| \\
			& \leq \bigg| J(t,\boldsymbol{x},y) - \widehat{U}(t,\boldsymbol{x},y) \bigg| + \bigg| \widehat{U}(t,\boldsymbol{x},y) - \mathbb{E}\Big( U_T\big(\boldsymbol{\hat{X}}(T)\big) \Big| \boldsymbol{\hat{X}}(t) = \boldsymbol{x}, Y(t) = y \Big) \bigg| \\
			& \leq c(T-t)^2 f(\boldsymbol{x}),
		\end{align*}
		for $ (t,\boldsymbol{x},y) \in (T-\varepsilon,T)\times (\mathbb{R}^+)^n \times \mathbb{R} $.
	\end{proof}

\subsection{Portfolio optimization on a finite time horizon}
\label{sec5}

In this subsection, we approximate the value function for all times $ t \in [0,T] $. Using this approximation with optimal portfolio \eqref{pi_hat},  we generate a close-to-optimal portfolio on $ [0, T] $. To start, we partition the interval $ [0, T] $ into $ m $ subintervals: $ \{ 0 = t_0 < t_1 < \cdots < t_{m-1} < t_m = T \} $. For $ t_k \leq t \leq t_{k+1} $, $ k = 0, \cdots, m-1 $, the approximation scheme is given by 
	\begin{align} \label{approx}
		\widehat{U}(t, \boldsymbol{x}, y) := \widehat{U}(t_{k+1}, \boldsymbol{x}, y) + (t_{k+1} - t)\mathcal{H}\big(\widehat{U}(t_{k+1}, \boldsymbol{x}, y)\big),
	\end{align}
	where
	\begin{align*}
		\mathcal{H}\big(\widehat{U}(t_{k+1}, \boldsymbol{x}, y)\big) = \
		& \sum_{i=1}^{n}  \lambda_i(y) \sigma_i(y) \widehat{\pi}_i(t_{k+1}, \boldsymbol{x}, y) \widehat{U}_{x_i}(t_{k+1}, \boldsymbol{x}, y) 
		+ b(y)\widehat{U}_y(t_{k+1}, \boldsymbol{x}, y) \\
		& + \frac{1}{2} \sum_{i,j=1}^{n} \rho_{ij} \sigma_i(y) \sigma_j(y) \widehat{\pi}_i(t_{k+1}, \boldsymbol{x}, y)  \widehat{\pi}_j(t_{k+1}, \boldsymbol{x}, y) \widehat{U}_{x_ix_j}(t_{k+1}, \boldsymbol{x}, y) \\
		& + a(y) \sum_{i=1}^{n} \omega_i \sigma_i(y) \widehat{\pi}_i(t_{k+1}, \boldsymbol{x}, y)  \widehat{U}_{x_iy}(t_{k+1}, \boldsymbol{x}, y)
		+ \frac{1}{2}a^2(y) \widehat{U}_{yy}(t_{k+1}, \boldsymbol{x}, y),
	\end{align*}
	with
	\begin{multline*}
		\widehat{\pi}_i (t_{k+1}, \boldsymbol{x}, y)
		= \
		 \frac{1}{2 \sigma_i(y) \widehat{U}_{x_ix_i}(t_{k+1}, \boldsymbol{x}, y)} 
		\Bigg[
		-3\Big[\lambda_i(y) \widehat{U}_{x_i}(t_{k+1}, \boldsymbol{x}, y) + \omega_i a(y) \widehat{U}_{x_iy}(t_{k+1}, \boldsymbol{x}, y)\Big] \\
		+  \sum_{k=1}^n 
		\rho_{ik} \Big[\lambda_k(y) \widehat{U}_{x_k}(t_{k+1}, \boldsymbol{x}, y) + \omega_k a(y) \widehat{U}_{x_ky}(t_{k+1}, \boldsymbol{x}, y)\Big]
		\Bigg].
	\end{multline*}
	By substituting the approximation of value function \eqref{approx} into optimal portfolio \eqref{pi_hat}, the corresponding close-to-optimal portfolio is
	\begin{multline*}
		\widetilde{\pi}_i (t, \boldsymbol{x}, y)
		= \frac{1}{2\sigma_i(y) \widehat{U}_{x_ix_i}(t, \boldsymbol{x}, y)} 
		\Bigg[
		-3\Big[\lambda_i(y) \widehat{U}_{x_i}(t, \boldsymbol{x}, y) + \omega_i a(y) \widehat{U}_{x_i y}(t, \boldsymbol{x}, y)\Big] \\
		+  \sum_{k=1}^n 
		\rho_{ik} \Big[\lambda_k(y) \widehat{U}_{x_k}(t, \boldsymbol{x}, y) + \omega_k a(y) \widehat{U}_{x_k y}(t, \boldsymbol{x}, y)\Big]
		\Bigg].
	\end{multline*}

\section{Numerical examples for approximating Merton-like value function}

In this section, we numerically analyze the accuracy of our approximate solution \eqref{U_hat} of the HJB equation \eqref{HJB},  and compare with respect to a Merton-like value function. Merton approximation to the optimal portfolio is used as a benchmark for its simplicity and wide acceptance in financial literature. However, here we consider a multi-dimensional case. We set $ n=2 $ for convenience. Referring to Section 6.2 in \cite{Kumar}, we choose a naive Merton-like value function by substituting $ \|\boldsymbol{x}\| $ for the $ x $ as the benchmark:
\begin{align*}
	U^{Mer}(\boldsymbol{x}) = -e^{-0.0001569674298} \frac{1}{2 \|\boldsymbol{x}\|^2} = - \frac{0.499922}{\|\boldsymbol{x}\|^2}.
\end{align*}
Here we set $ \lambda_1^2(y) =  \lambda_2^2(y) = 0.0002354511446 $ at $ y = 27.9345 $ (same as in \cite{Kumar,LinSG}).
For comparison, we choose the approximation of value function given by the equation (5.2) in \cite{Kumar} with substituting $ \|\boldsymbol{x}\| $ for the $ x $, and denoting it as $ U^C(t, \boldsymbol{x}, y) $.

Based on the above discussion, we set $ U_T(\boldsymbol{x}) = - \frac{1}{2} \|\boldsymbol{x}\|^{-2} $, $ \rho_{12} = \rho_{21} = 0.5241 $, and $ \lambda_1^2(y) = \lambda_2^2(y) = 0.0002354511446 $ at $ y = 27.9345 $ for all $ t $.
We then calculate our approximations of value function at times $ t = 1.5 $ and $ t = 1.9 $, which are close to terminal time $ T=2 $.
The results with respective error are summarized in Table 1, 2, and 3, and are graphed in Fugures 1, 2, 3, and 4 respectively.

\begin{table}[H]
	\caption{$t = 1.5$ and $ 1.9 $, $T = 2$. The comparison value function $ U^C(t, \boldsymbol{x}, y) $, our approximation value function $  \widehat{U}(t, \boldsymbol{x}, y) $, and respective error functions comparing with Merton-like value function $ U^{Mer}(\boldsymbol{x}) = - 0.499922\|\boldsymbol{x}\|^{-2} $.}
	\begin{center}
		\begin{tabular}{c c c c c c c c c c c}
			\hline
			$t$ & $T$ & $ U^C(t, \boldsymbol{x}, y) $ & $ \widehat{U}(t, \boldsymbol{x}, y)  $ & $ \big|U^{Mer} - U^C\big|   $ & $ \big|U^{Mer} - \widehat{U}\big| $   \\
			\hline \\
			1.5 & 2 & $ \approx -\frac{0.484689}{\|\boldsymbol{x}\|^2} $ & $ \approx -\frac{0.499975}{\|\boldsymbol{x}\|^2} $ & $ \approx \frac{0.015233}{\|\boldsymbol{x}\|^2} $ & $ \approx \frac{0.000053}{\|\boldsymbol{x}\|^2} $  \\
			\\
			1.9 & 2 & $ \approx -\frac{0.496938}{\|\boldsymbol{x}\|^2} $ & $ \approx -\frac{0.499995}{\|\boldsymbol{x}\|^2} $ & $ \approx \frac{0.002984}{\|\boldsymbol{x}\|^2} $ & $ \approx \frac{0.000073}{\|\boldsymbol{x}\|^2} $  \\
			\\
			\hline 
		\end{tabular}
	\end{center}
\end{table}

The approximation errors are of the expected order. In fact, with respect to the Merton-like value function (benchmark), the improvement for the proposed technique is significant compared to \cite{Kumar}. This is illustrated in the last two columns of Table 1.  From  Tables 2 and 3, it is clear that our approximating utility values have much smaller percentage error values compared to the utility values in \cite{Kumar}.  Consequently, it is difficult to distinguish between our approximating value function and the Merton-like value function (benchmark), see Figures 1 and 2. We can also observe those when time interval $ (T-t) $ is shortened from $ 0.5 $ to $ 0.1 $, in Figures 3 and 4.

\begin{table}[H]
	\caption{$t = 1.5$, $T = 2$. With some total wealth ($ \|\boldsymbol{x}\| $), the Merton-like utility values ($ U^{Mer} $) are calculated against the comparison utility values ($ U^C $) and our approximation utility values ($ \widehat{U} $). The corresponding percentage error values are also summarized.}
	\begin{center}
		\begin{tabular}{c c c c c c}
			\hline 
			$ x $ & $ U^{Mer} $ & $ U^C $ & $ \widehat{U} $ & $ \big|\frac{U^{Mer} - U^C}{U^{Mer}}\big| \times 100\% $ & $ \big|\frac{U^{Mer} - \widehat{U}}{U^{Mer}}\big| \times 100\% $ \\
			\hline \\
			$ 0.4 $ & $ -3.124509 $ & $ -3.029306 $ & $ -3.124842 $ & $ 3.0470\% $ & $ 0.0107\% $ \\
			\\
			$ 0.5 $ & $ -1.999686 $ & $ -1.938756 $ & $ -1.999899 $ & $ 3.0470\% $ & $ 0.0107\% $ \\
			\\
			$ 0.6 $ & $ -1.388671 $ & $ -1.346358 $ & $ -1.388818 $ & $ 3.0470\% $ & $ 0.0106\% $ \\
			\\
			$ 0.7 $ & $ -1.020248 $ & $ -0.989161 $ & $ -1.020356 $ & $ 3.0470\% $ & $ 0.0106\% $ \\
			\\
			$ 0.8 $ & $ -0.781127 $ & $ -0.757327 $ & $ -0.781210 $ & $ 3.0469\% $ & $ 0.0106\% $ \\
			\\
			$ 0.9 $ & $ -0.617187 $ & $ -0.598381 $ & $ -0.617253 $ & $ 3.0471\% $ & $ 0.0107\% $ \\
			\\
			$ 1.0 $ & $ -0.499922 $ & $ -0.484689 $ & $ -0.499975 $ & $ 3.0471\% $ & $ 0.0106\% $ \\
			\\
			$ 1.1 $ & $ -0.413158 $ & $ -0.400569 $ & $ -0.413202 $ & $ 3.0470\% $ & $ 0.0106\% $ \\
			\\
			$ 1.2 $ & $ -0.347168 $ & $ -0.336590 $ & $ -0.347205 $ & $ 3.0469\% $ & $ 0.0107\% $ \\
			\\
			$ 1.3 $ & $ -0.295812 $ & $ -0.286798 $ & $ -0.295843 $ & $ 3.0472\% $ & $ 0.0105\% $ \\
			\\
			$ 1.4 $ & $ -0.255062 $ & $ -0.247290 $ & $ -0.255089 $ & $ 3.0471\% $ & $ 0.0106\% $ \\
			\\
			$ 1.5 $ & $ -0.222187 $ & $ -0.215417 $ & $ -0.222211 $ & $ 3.0470\% $ & $ 0.0108\% $ \\
			\\
			$ 1.6 $ & $ -0.195282 $ & $ -0.189332 $ & $ -0.195303 $ & $ 3.0469\% $ & $ 0.0108\% $ \\
			\\
			\hline 
		\end{tabular}
	\end{center}
\end{table}

\begin{table}[H]
	\caption{$t = 1.9$, $T = 2$. With some total wealth ($ \|\boldsymbol{x}\| $), the Merton-like utility values ($ U^{Mer} $) are calculated against the comparison utility values ($ U^C $) and our approximation utility values ($ \widehat{U} $). The corresponding percentage error values are also summarized.}
	\begin{center}
		\begin{tabular}{c c c c c c}
			\hline 
			$ x $ & $ U^{Mer} $ & $ U^C $ & $ \widehat{U} $ & $ \big|\frac{U^{Mer} - U^C}{U^{Mer}}\big| \times 100\% $ & $ \big|\frac{U^{Mer} - \widehat{U}}{U^{Mer}}\big| \times 100\% $ \\
			\hline \\
			$ 0.4 $ & $ -3.124509 $ & $ -3.029306  $ & $ -3.124968 $ & $ 3.0470\% $ & $ 0.0147\% $ \\
			\\
			$ 0.5 $ & $ -1.999686 $ & $ -1.938756 $ & $ -1.999980 $ & $ 3.0470\% $ & $ 0.0147\% $ \\
			\\
			$ 0.6 $ & $ -1.388671 $ & $ -1.346358 $ & $ -1.388875 $ & $ 3.0470\% $ & $ 0.0147\% $ \\
			\\
			$ 0.7 $ & $ -1.020248 $ & $ -0.989161 $ & $ -1.020398 $ & $ 3.0470\% $ & $ 0.0147\% $ \\
			\\
			$ 0.8 $ & $ -0.781127 $ & $ -0.757327 $ & $ -0.781242 $ & $ 3.0469\% $ & $ 0.0147\% $ \\
			\\
			$ 0.9 $ & $ -0.617187 $ & $ -0.598381 $ & $ -0.617278 $ & $ 3.0471\% $ & $ 0.0147\% $ \\
			\\
			$ 1.0 $ & $ -0.499922 $ & $ -0.484689 $ & $ -0.499995 $ & $ 3.0471\% $ & $ 0.0146\% $ \\
			\\
			$ 1.1 $ & $ -0.413158 $ & $ -0.400569 $ & $ -0.413219 $ & $ 3.0470\% $ & $ 0.0148\% $ \\
			\\
			$ 1.2 $ & $ -0.347168 $ & $ -0.336590 $ & $ -0.347219 $ & $ 3.0469\% $ & $ 0.0147\% $ \\
			\\
			$ 1.3 $ & $ -0.295812 $ & $ -0.286798 $ & $ -0.295855 $ & $ 3.0472\% $ & $ 0.0145\% $ \\
			\\
			$ 1.4 $ & $ -0.255062 $ & $ -0.247290 $ & $ -0.255099 $ & $ 3.0471\% $ & $ 0.0145\% $ \\
			\\
			$ 1.5 $ & $ -0.222187 $ & $ -0.215417 $ & $ -0.222220 $ & $ 3.0470\% $ & $ 0.0149\% $ \\
			\\
			$ 1.6 $ & $ -0.195282 $ & $ -0.189332 $ & $ -0.195311 $ & $ 3.0469\% $ & $ 0.0149\% $ \\
			\\
			\hline 
		\end{tabular}
	\end{center}
\end{table}

\begin{figure}[H]
	\centering
	\begin{minipage}{\textwidth}
		\centering
		\includegraphics[width=0.6\textwidth]{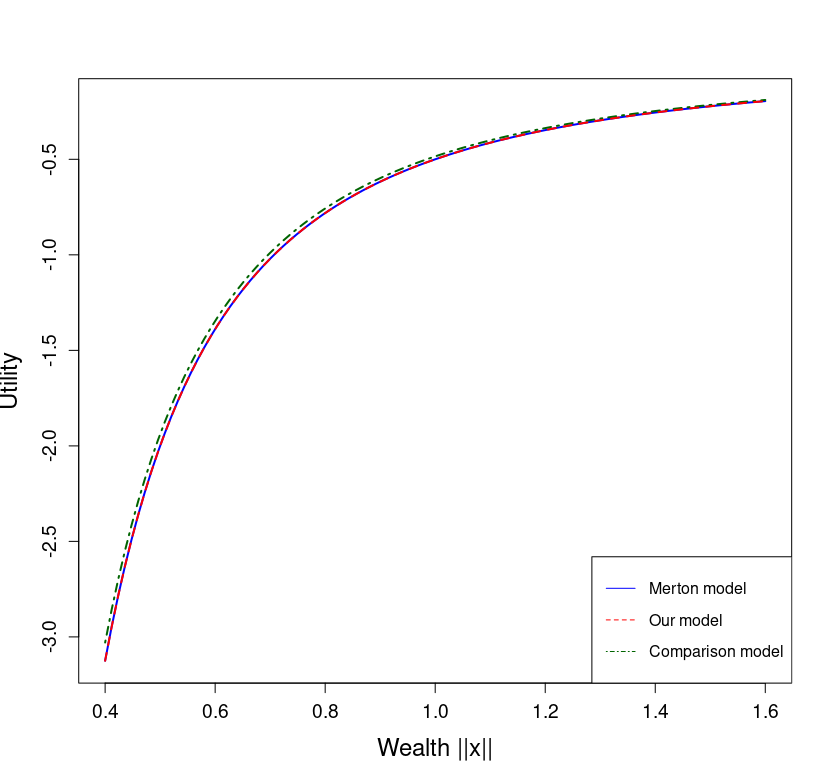}
		\caption{($t = 1.5$, $T = 2$) The Merton-like value function $ U^{Mer} $ is plotted against our approximation value function $  \widehat{U} $ and the comparison value function $ U^C $. It is difficult to distinguish between the Merton-like value function (benchmark) and our approximation value function.}
	\end{minipage}
	\begin{minipage}{\textwidth}
		\centering
		\includegraphics[width=0.6\textwidth]{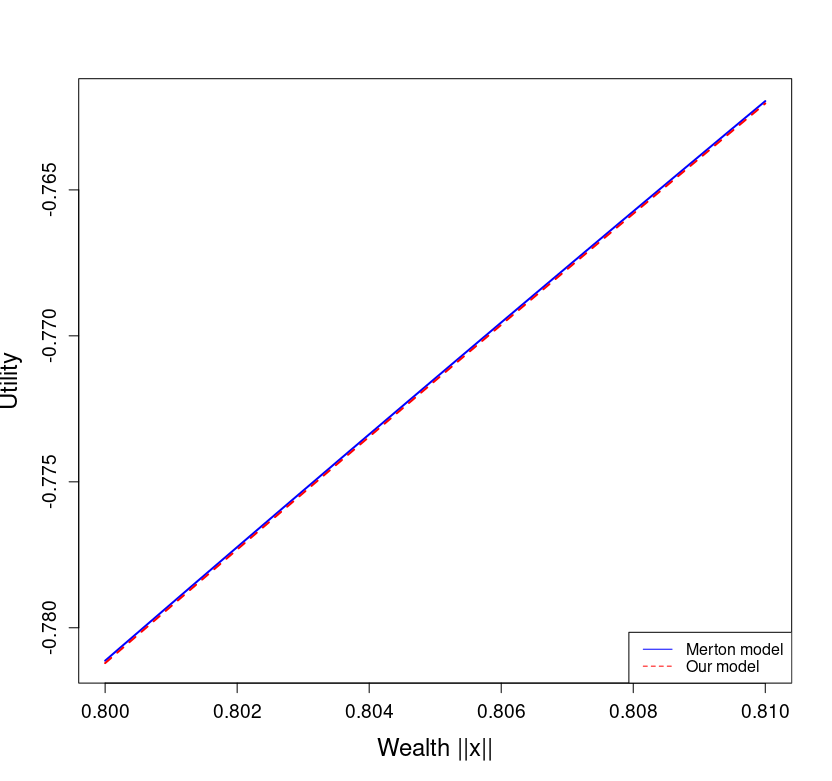}
		\caption{($t = 1.5$, $T = 2$) When Figure 1 is
			zoomed in over a shorter wealth interval, difference between our approximation value function and the Merton-like value function (benchmark) is apparent.}
	\end{minipage}
\end{figure}

\begin{figure}[H]
	\centering
	\begin{minipage}{\textwidth}
		\centering
		\includegraphics[width=0.6\textwidth]{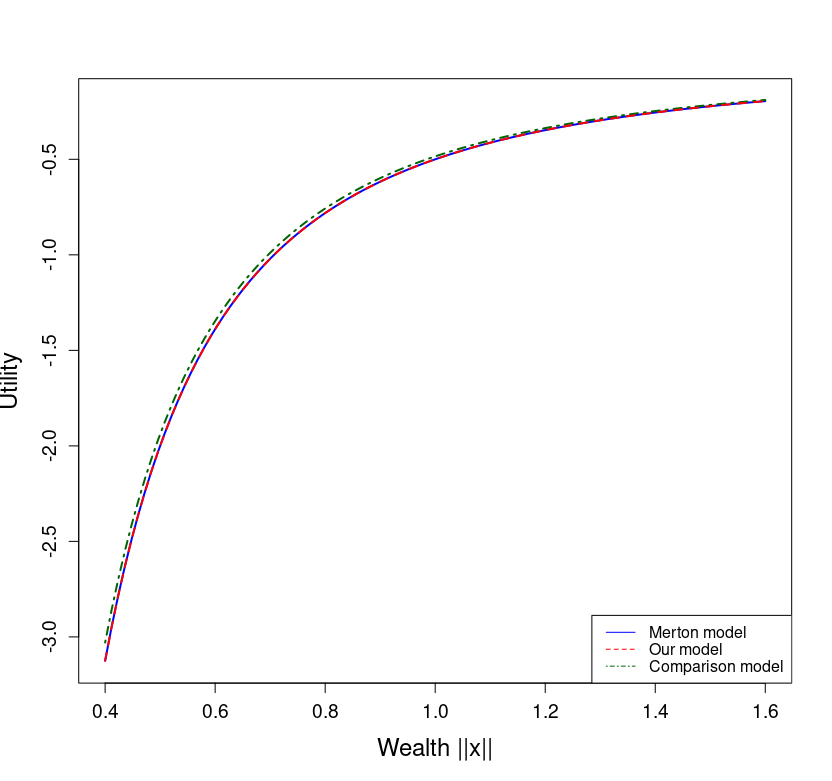}
		\caption{($t = 1.9$, $T = 2$) The Merton-like value function $ U^{Mer} $ is plotted against our approximation value function $  \widehat{U} $ and the comparison value function $ U^C $. It is difficult to distinguish between the Merton-like value function (benchmark) and our approximation value function.}
	\end{minipage}
	\begin{minipage}{\textwidth}
		\centering
		\includegraphics[width=0.6\textwidth]{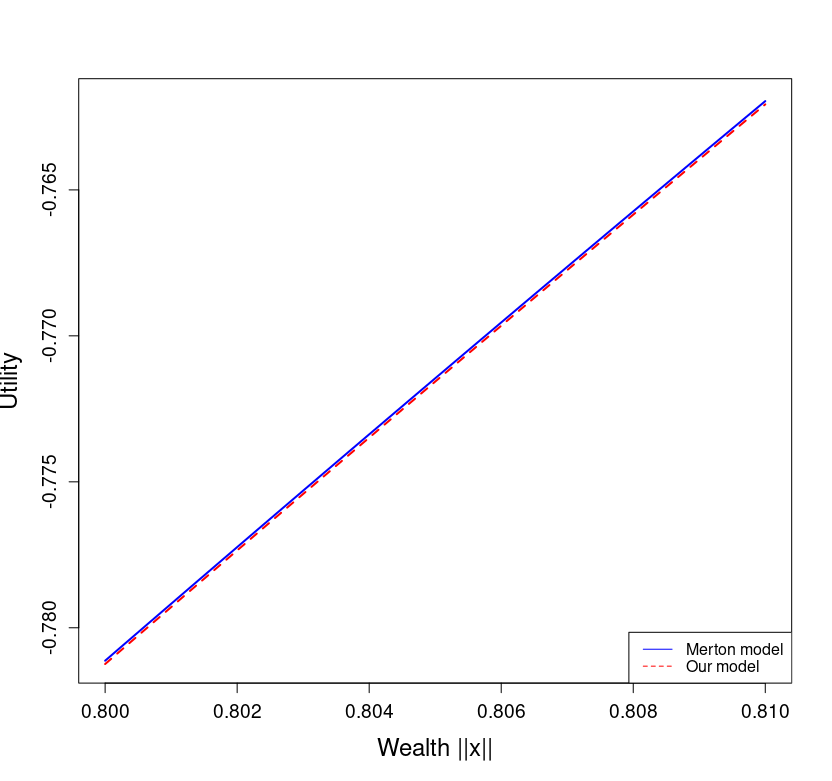}
		\caption{($t = 1.9$, $T = 2$) When Figure 3 is
			zoomed in over a shorter wealth interval, difference between our approximation value function and the Merton-like value function (benchmark) is apparent.}
	\end{minipage}
\end{figure}

\section{Conclusion}
\label{sec6}
In this paper, we consider the finite horizon portfolio optimization in a stochastic volatility market model driven by $ n $-dimensional Brownian motions. 
The value function is approximated using the polynomial expansion method with respect to time to the horizon $ (T - t) $. 
We obtain an approximate solution for the value function and optimal investment strategy. 
It is shown that our $ n $-dimensional model includes all of the advantages of its corresponding $ 1 $-dimensional model in \cite{Kumar}.	
For the $ 2 $-dimensional case, as an example, it is shown that our first-order approximations of the value function and optimal investment strategy perform better than the existing models such as \cite{Kumar}. 

Based on our approximation, we also generate a close-to-optimal portfolio near the time to horizon $ (T - t) $. 
We provide an approximation scheme to the value function for all times $ t \in [0, T ] $ and generate the close-to-optimal portfolio on $ [0, T] $. The accuracy of such approximation will be accomplished by a similar procedure used in Section \ref{sec3} and will be rigorously proved in a sequel of this work.

\appendix

\section{Appendix: Proofs of some technical lemmas}
\label{apendixA}

	\begin{lemma} \label{A.0}
		Suppose that $ U^{(2)}(\boldsymbol{x},y) $ is given by \eqref{U^2}. If all the terms of $ U^{(2)}(\boldsymbol{x},y) $ are enumerated as $ u_1^{(2)}(\boldsymbol{x},y) , \cdots , u_m^{(2)}(\boldsymbol{x},y) $, i.e.\ $ U^{(2)}(\boldsymbol{x},y) = \sum_{i=1}^m u_i^{(2)}(\boldsymbol{x},y) $, then
		\begin{align*}
			u_i^{(2)}(\boldsymbol{x},y) \sim f(\boldsymbol{x})
		\end{align*}
		for $ 1\leq i \leq m $, where $ f(\boldsymbol{x}) = 1 $ under Case 1 of Assumption \ref{Assume_UT} and $ f(\boldsymbol{x}) = \|\boldsymbol{x}\|^{1-\alpha} + \|\boldsymbol{x}\|^{1-\beta} $ under Case 2 of Assumption \ref{Assume_UT}.
	\end{lemma}
	\begin{proof}
		We denote $ U^{(0)} = U_T(\boldsymbol{x}) $.
		From \eqref{U^2}, we observe that $ U^{(2)} $ is a linear combination of the following terms:
		\begin{align*}
			\frac{\big[U^{(0)}_{x_i}\big]^2}{U^{(0)}_{x_ix_i}}, \
			\frac{\big[U^{(0)}_{x_i}\big]^3U^{(0)}_{x_ix_ix_i}}{\big[U^{(0)}_{x_ix_i}\big]^3}, \
			\frac{\big[U^{(0)}_{x_i}\big]^4U^{(0)}_{x_ix_ix_ix_i}}{\big[U^{(0)}_{x_ix_i}\big]^4}, \text{ and }
			\frac{\big[U^{(0)}_{x_i}\big]^4 \big[U^{(0)}_{x_ix_ix_i}\big]^2}{\big[U^{(0)}_{x_ix_i}\big]^5},\ i = 1, \dots, n.
		\end{align*}
		We let $ f(\boldsymbol{x}) = 1 $ under Case 1 of Assumption \ref{Assume_UT}. Then, for $ i = 1, \dots, n $,
		\begin{align*}
			& \lim_{\|\boldsymbol{x}\| \rightarrow \infty} \frac{\big[U^{(0)}_{x_i}\big]^2}{U^{(0)}_{x_ix_i}} \Big/ f(\boldsymbol{x})
			= -1,
			\qquad \qquad \ \ 
			\lim_{\|\boldsymbol{x}\| \rightarrow \infty} \frac{\big[U^{(0)}_{x_i}\big]^3U^{(0)}_{x_ix_ix_i}}{\big[U^{(0)}_{x_ix_i}\big]^3} \Big/ f(\boldsymbol{x}) = -2, \\
			& \lim_{\|\boldsymbol{x}\| \rightarrow \infty} \frac{\big[U^{(0)}_{x_i}\big]^4U^{(0)}_{x_ix_ix_ix_i}}{\big[U^{(0)}_{x_ix_i}\big]^4} \Big/ f(\boldsymbol{x}) = -6,
			\quad
			\lim_{\|\boldsymbol{x}\| \rightarrow \infty} \frac{\big[U^{(0)}_{x_i}\big]^4 \big[U^{(0)}_{x_ix_ix_i}\big]^2}{\big[U^{(0)}_{x_ix_i}\big]^5} \Big/ f(\boldsymbol{x}) = -4.
		\end{align*}
		We let $ f(\boldsymbol{x}) = \|\boldsymbol{x}\|^{1-\alpha} + \|\boldsymbol{x}\|^{1-\beta} $ under Case 2 of Assumption \ref{Assume_UT}. Then, for $ i = 1, \dots, n $,
		\begin{align*}
			\lim_{\|\boldsymbol{x}\| \rightarrow \infty} \frac{\big[U^{(0)}_{x_i}\big]^2}{U^{(0)}_{x_ix_i}} \Big/ f(\boldsymbol{x})
			& = \lim_{\|\boldsymbol{x}\| \rightarrow \infty}
			\frac{\big[c_1\|\boldsymbol{x}\|^{-\alpha}+c_2\|\boldsymbol{x}\|^{-\beta}\big]^2}
			{\big[-c_1\alpha\|\boldsymbol{x}\|^{-\alpha}-c_2\beta\|\boldsymbol{x}\|^{-\beta}\big] \big[\|\boldsymbol{x}\|^{-\alpha} + \|\boldsymbol{x}\|^{-\beta}\big]} \\
			& =
			\left\{
			\begin{array}{ll}
				-c_1/\alpha, \qquad\qquad\ \text{if } \alpha < \beta \\
				-(c_1+c_2)/(2\alpha), \text{ if } \alpha = \beta \\
				-c_2/\beta, \qquad\qquad\ \text{if } \alpha > \beta
			\end{array}
			\right.
			.
		\end{align*}
		Similarly, 
		\begin{align*}
			& \lim_{\|\boldsymbol{x}\| \rightarrow \infty} \frac{\big[U^{(0)}_{x_i}\big]^3 U^{(0)}_{x_ix_ix_i}}{\big[U^{(0)}_{x_ix_i}\big]^3} \Big/ f(\boldsymbol{x})
			=
			\left\{
			\begin{array}{ll}
				-c_1(\alpha+1)/\alpha^2, \qquad\qquad\ \text{if } \alpha < \beta \\
				-(c_1+c_2)(\alpha+1)/(2\alpha^2), \text{ if } \alpha = \beta \\
				-c_2(\beta+1)/\beta^2, \qquad\qquad\ \text{if } \alpha > \beta
			\end{array}
			\right.			
			, \\
			& \lim_{\|\boldsymbol{x}\| \rightarrow \infty} \frac{\big[U^{(0)}_{x_i}\big]^4U^{(0)}_{x_ix_ix_ix_i}}{\big[U^{(0)}_{x_ix_i}\big]^4} \Big/ f(\boldsymbol{x})
			=
			\left\{
			\begin{array}{ll}
				- c_1(\alpha+1)(\alpha+2)/\alpha^3, \qquad\qquad\ \text{if } \alpha < \beta \\
				- (c_1+c_2)(\alpha+1)(\alpha+2)/(2\alpha^3), \text{ if } \alpha = \beta \\
				- c_2(\beta+1)(\beta+2)/\beta^3, \qquad\qquad\ \text{if } \alpha > \beta
			\end{array}
			\right.			
			, \\
			& \lim_{\|\boldsymbol{x}\| \rightarrow \infty} \frac{\big[U^{(0)}_{x_i}\big]^4 \big[U^{(0)}_{x_ix_ix_i}\big]^2}{\big[U^{(0)}_{x_ix_i}\big]^5} \Big/ f(\boldsymbol{x})
			=
			\left\{
			\begin{array}{ll}
				- c_1(\alpha+1)^2/\alpha^3, \qquad\qquad\ \text{if } \alpha < \beta \\
				- (c_1+c_2)(\alpha+1)^2/(2\alpha^3), \text{ if } \alpha = \beta \\
				- c_2(\beta+1)^2/\beta^3, \qquad\qquad\ \text{if } \alpha > \beta
			\end{array}
			\right.			
			.
		\end{align*}
		By definition, for $ i = 1, \dots, n $, we have
		\begin{align*}
			\frac{\big[U^{(0)}_{x_i}\big]^2}{U^{(0)}_{x_ix_i}} \sim f(\boldsymbol{x}), 
			\frac{\big[U^{(0)}_{x_i}\big]^3U^{(0)}_{x_ix_ix_i}}{\big[U^{(0)}_{x_ix_i}\big]^3} \sim f(\boldsymbol{x}), 
			\frac{\big[U^{(0)}_{x_i}\big]^4U^{(0)}_{x_ix_ix_ix_i}}{\big[U^{(0)}_{x_ix_i}\big]^4} \sim f(\boldsymbol{x}), 
			\text{and}
			\frac{\big[U^{(0)}_{x_i}\big]^4 \big[U^{(0)}_{x_ix_ix_i}\big]^2}{\big[U^{(0)}_{x_ix_i}\big]^5} \sim f(\boldsymbol{x}).
		\end{align*}
		Finally, by the sum law and constant multiple law of limit, we have $ u_i^{(2)}(\boldsymbol{x},y) \sim f(\boldsymbol{x}) $ for $ 1\leq i \leq m $.
	\end{proof}
	\begin{lemma} \label{A.1}
		Let $ x_i = X_i(t) $ be the $ i $-th wealth process defined by \eqref{Xt}, $ n = 1, \cdots, n $, under any admissible portfolio. Denote $ \boldsymbol{x} = (x_1,\dots,x_n) $ and $ \boldsymbol{X}(\cdot) = \big(X_1(\cdot),\dots,X_n(\cdot)\big) $. Let $ g(\boldsymbol{x}) = \ln \|\boldsymbol{x}\| +1 $ under Case 1 of Assumption \ref{Assume_UT}, and $ g(\boldsymbol{x}) = \|\boldsymbol{x}\|^{1-\alpha} + \|\boldsymbol{x}\|^{1-\beta} $, $ \alpha, \beta > 0 $ and $ \alpha, \beta \neq 1 $, under Case 2 of Assumption \ref{Assume_UT}. Then $ \big\{ g\big(\boldsymbol{X}(t_n)\big) \big\}_{n=1}^\infty $ is uniformly bounded by an integrable random variable, where $ \{t_n \}_{n=1}^\infty \subset [t, T] $ is a sequence of stopping times s.t. $ t_n \leq t_{n+1} $ and $ t_n \rightarrow T $.
	\end{lemma}
	\begin{proof}
		We denote $ \sigma_i(t) = \sigma_i\big(Y(t)\big) $, $ \pi_i(t) = \pi_i\big(t, \boldsymbol{X}(t), Y(t)\big) $, and $ \lambda_i(t) = \lambda_i\big(Y(t)\big) $, and let $ \tau \in [t, T] $.  For $ g(\boldsymbol{x}) = \ln\|\boldsymbol{x}\| +1 $, it is identical to show that $ \big\{ \ln \|\boldsymbol{X}(t_n)\| \big\}_{n=1}^\infty $ is uniformly bounded by an integrable random variable. We apply It\^{o}'s formula to $ \ln\|\boldsymbol{X}(\tau)\| $ and obtain
		\begin{align*}
			\ln\|\boldsymbol{X}(\tau)\| = \
			&  \ln \|\boldsymbol{x}\| 
			+ \sum_{i=1}^n
			\int_t^\tau \bigg[ \frac{\sigma_i(s) \pi_i(s) \lambda_i(s)}{X_i(s)}
				- \frac{\sigma^2_i(s) \pi^2_i(s)}{2X_i^2(s)} \bigg] ds
			+ \sum_{i=1}^n 
			\int_t^\tau  \frac{\sigma_i(s) \pi_i(s)}{X_i(s)} dW_i(s).
		\end{align*}
		Taking expectation, we choose some constant $ c_1 $ such that 
		\begin{align*}
			\mathbb{E} \big( \ln\|\boldsymbol{X}(\tau)\| \big) \leq\
			&  c_1 \Bigg[
			1 +  \mathbb{E} \bigg( \sum_{i=1}^n \int_0^T \frac{\sigma_i^2(s) \pi_i^2(s) }{X_i^2(s)} ds \bigg)
			+ \mathbb{E} \bigg( \sup_\tau \bigg[ \sum_{i=1}^n \int_t^\tau  \frac{\sigma_i(s) \pi_i(s)}{X_i(s)} dW_i(s) \bigg]^2 \bigg)
			\Bigg].
		\end{align*}
		By Doob's martingale maximal inequalities, we have
		\begin{align*}
			\mathbb{E} \big( \ln\|\boldsymbol{X}(\tau)\| \big) \leq \
			& c_1 \Bigg[
			1 +  \mathbb{E} \bigg( \sum_{i=1}^n \int_0^T \frac{\sigma_i^2(s) \pi_i^2(s) }{X_i^2(s)} ds \bigg)
			+ 4 \mathbb{E} \bigg( \bigg[ \sum_{i=1}^n \int_t^\tau  \frac{\sigma_i(s) \pi_i(s)}{X_i(s)} dW_i(s) \bigg]^2 \bigg)
			\Bigg].
		\end{align*}
		Then by It\^{o} isometries, we have
		\begin{align*}
			\mathbb{E} \big( \ln\|\boldsymbol{X}(\tau)\| \big) \leq \
			& c_1 \Bigg[
			1 +  5\mathbb{E} \bigg( \sum_{i=1}^n \int_0^T \frac{\sigma_i^2(s) \pi_i^2(s) }{X_i^2(s)} ds \bigg)
			\Bigg].
		\end{align*}
		Finally, by Definition \ref{def_1}, the right-hand side of above inequality is finite, which serves as the uniform bound for $ \mathbb{E} \big( \ln\|\boldsymbol{X}(\tau)\| \big) $, $ \tau \in [t, T] $. This implies the $ \big\{ \ln\|\boldsymbol{X}(t_n)\| \big\}_{n=1}^\infty $ is uniformly bounded by an integrable random variable.

		For $ g(\boldsymbol{x}) = \|\boldsymbol{x}\|^{1-\alpha} + \|\boldsymbol{x}\|^{1-\beta} $, it is identical to show that $ \big\{ \|\boldsymbol{X}(t_n)\|^{1-r} \big\}_{n=1}^\infty $, $ r > 0 $ and $ r \neq 1, $ is uniformly bounded by an integrable random variable. We apply It\^{o}'s formula to $ \|\boldsymbol{X}(\tau)\|^{1-r} $ and obtain
		\begin{multline*}
			\|\boldsymbol{X}(\tau)\|^{1-r} \\
			= \|\boldsymbol{x}\|^{1-r} + [1-r] \Bigg[ 
			\sum_{i=1}^n
			\int_t^\tau \bigg[ \frac{\sigma_i(s) \pi_i(s) \lambda_i(s)}{X_i^r(s)}
			- \frac{r \sigma^2_i(s) \pi^2_i(s)}{2X_i^{r+1}(s)} \bigg] ds
			+ \sum_{i=1}^n 
			\int_t^\tau  \frac{\sigma_i(s) \pi_i(s)}{X_i^r(s)} dW_i(s)
			\Bigg].
		\end{multline*}
		Taking expectation, we choose some constant $ c_2 $ such that 
		\begin{align*}
			\mathbb{E} \big(\|\boldsymbol{X}(\tau)\|^{1-r}\big)
			\leq c_2 \Bigg[ 1+ \mathbb{E} \bigg(
			\sum_{i=1}^n
			\int_0^T \frac{\sigma^2_i(s) \pi^2_i(s)}{X_i^{2r}(s)} ds \bigg)
			+ \mathbb{E} \bigg(  \sup_\tau \bigg[ \sum_{i=1}^n 
			\int_t^\tau  \frac{\sigma_i(s) \pi_i(s)}{X_i^r(s)} dW_i(s) \bigg]^2 \bigg)
			\Bigg].
		\end{align*}
		By Doob's martingale maximal inequalities, we have
		\begin{align*}
			\mathbb{E} \big(\|\boldsymbol{X}(\tau)\|^{1-r}\big)
			\leq c_2 \Bigg[ 1+ \mathbb{E} \bigg(
			\sum_{i=1}^n
			\int_0^T \frac{\sigma^2_i(s) \pi^2_i(s)}{X_i^{2r}(s)} ds \bigg)
			+ 4 \mathbb{E} \bigg(\bigg[ \sum_{i=1}^n 
			\int_t^\tau  \frac{\sigma_i(s) \pi_i(s)}{X_i^r(s)} dW_i(s) \bigg]^2 \bigg)
			\Bigg].
		\end{align*}
		Then by It\^{o} isometries, we have
		\begin{align*}
			\mathbb{E} \big(\|\boldsymbol{X}(\tau)\|^{1-r}\big)
			\leq c_2 \Bigg[ 1+ 5 \mathbb{E} \bigg(
			\sum_{i=1}^n
			\int_0^T \frac{\sigma^2_i(s) \pi^2_i(s)}{X_i^{2r}(s)} ds \bigg)
			\Bigg].
		\end{align*}
	Finally, by Definition \ref{def_1}, the right-hand side of above inequality is finite, which serves as the uniform bound for $ \mathbb{E} \big(\|\boldsymbol{X}(\tau)\|^{1-r}\big) $, $ \tau \in [t, T] $. This implies the $ \big\{ \|\boldsymbol{X}(t_n)\|^{1-r} \big\}_{n=1}^\infty $, $ r > 0 $, $ r \neq 1, $ is uniformly bounded by an integrable random variable.
	\end{proof}
	
	\begin{lemma} \label{A.2}
		For $ i = 1,\dots,n $, if the wealth $ \widehat{X}_i(t) $ is defined by
		\begin{align} \label{A.2.1}
			d\widehat{X}_i(t) = \sigma_i\big(Y(t)\big) \widetilde{\pi}_i\big(t,\widehat{X}_1(t),\dots,\widehat{X}_n(t),Y(t)\big) \Big[ \lambda_i\big(Y(t)\big)dt + dW_i(t) \Big]
		\end{align}
		where
		\begin{multline} \label{A.2.2}
			\widetilde{\pi}_i (t,\boldsymbol{x},y) =\
			\frac{1}{2\sigma_i(y) \widehat{U}_{x_ix_i}(t,\boldsymbol{x},y)} 
			\Bigg[
			-3\Big[\lambda_i(y)\widehat{U}_{x_i}(t,\boldsymbol{x},y)+a(y) \omega_i \widehat{U}_{x_iy}(t,\boldsymbol{x},y)\Big] \\
			+  \sum_{k=1}^n 
			\rho_{ik} \Big[\lambda_k(y) \widehat{U}_{x_k}(t,\boldsymbol{x},y)+a(y) \omega_k \widehat{U}_{x_ky}(t,\boldsymbol{x},y)\Big]
			\Bigg]
		\end{multline}
		with $ \widehat{U}(t,\boldsymbol{x},y) $ given by \eqref{U_hat}, and $ \big(\widehat{X}_1(t),\dots,\widehat{X}_n(t)\big) = (x_1,\dots,x_n) = \boldsymbol{x}$, where $ x_i > 0 $ (for $i=1,\dots,n$), $ Y(t) = y $; then the optimal portfolio $ \widetilde{\pi}_i(t,\boldsymbol{x},y)  $ is admissible under Assumption \ref{Assume_UT}.
	\end{lemma}
	\begin{proof}
		We denote $ \sigma_i(t) = \sigma_i\big(Y(t)\big) $ and $ \widetilde{\pi}_i(t) = \widetilde{\pi}_i\big(t, \widehat{X}_1(t),\dots,\widehat{X}_n(t), Y(t)\big) $, $ i = 1,\dots,n $.
		Because $ \widetilde{\pi}_i(t) $ is continuous with respect to $ t $, therefore it is progressively measurable. 
		
		We denote $ U^{(0)} = U_T(\boldsymbol{x}) $. Referring to the proof of Lemma \ref{A.0}, from equation \eqref{U_hat}, we observe that $ \widehat{U}(t,\boldsymbol{x},y) $ is a linear combination of the terms that are equivalent to $ U^{(0)} $ and $ [U^{(0)}_{x_i}]^2 \big/ U^{(0)}_{x_ix_i} $, $ i = 1,\dots,n $.
		Then we observe $ \widetilde{\pi}_i (t,\boldsymbol{x},y) $ is a linear combination of the following terms:
		\begin{align*}
			\frac{U^{(0)}_{x_i}}{U^{(0)}_{x_ix_i}}, \
			\frac{\big[U^{(0)}_{x_i}\big]^2U^{(0)}_{x_ix_ix_i}}{\big[U^{(0)}_{x_ix_i}\big]^3}, \
			\frac{U^{(0)}_{x_ix_i}}{U^{(0)}_{x_ix_ix_i}}, \
			\frac{\big[U^{(0)}_{x_ix_i}\big]^3}{U^{(0)}_{x_i}\big[U^{(0)}_{x_ix_ix_i}\big]^2}, \
			\frac{U^{(0)}_{x_ix_ix_i}}{U^{(0)}_{x_ix_ix_ix_i}},
			\text{ and }
			\frac{\big[U^{(0)}_{x_ix_i}\big]^2}{U^{(0)}_{x_i}U^{(0)}_{x_ix_ix_ix_i}},\ i = 1, \dots, n.
		\end{align*}
		
		For $ i = 1,\dots,n $, we note $ \|\boldsymbol{x}\| \big/ x_i \rightarrow 1 $ as $ x_i \rightarrow \infty $. Then we have, under Case 1 of Assumption \ref{Assume_UT}
		\begin{align*}
			&\lim_{x_i \rightarrow \infty} \frac{U^{(0)}_{x_i}}{x_i U^{(0)}_{x_ix_i}} = -1,\
			\qquad \qquad
			\lim_{x_i \rightarrow \infty} \frac{\big[U^{(0)}_{x_i}\big]^2U^{(0)}_{x_ix_ix_i}}{x_i \big[U^{(0)}_{x_ix_i}\big]^3} = -2,\
			\lim_{x_i \rightarrow \infty} \frac{U^{(0)}_{x_ix_i}}{x_i U^{(0)}_{x_ix_ix_i}} = -\frac{1}{2},\\
			&\lim_{x_i \rightarrow \infty} \frac{\big[U^{(0)}_{x_ix_i}\big]^3}{x_i U^{(0)}_{x_i}\big[U^{(0)}_{x_ix_ix_i}\big]^2} = -\frac{1}{4},\
			\lim_{x_i \rightarrow \infty} \frac{U^{(0)}_{x_ix_ix_i}}{x_i U^{(0)}_{x_ix_ix_ix_i}} = -\frac{1}{3},\
			\quad \
			\lim_{x_i \rightarrow \infty} \frac{\big[U^{(0)}_{x_ix_i}\big]^2}{x_i U^{(0)}_{x_i}U^{(0)}_{x_ix_ix_ix_i}} = -\frac{1}{6};
		\end{align*}
		and under Case 2 of Assumption \ref{Assume_UT}
		\begin{align*}
			&\lim_{x_i \rightarrow \infty} \frac{U^{(0)}_{x_i}}{x_i U^{(0)}_{x_ix_i}} = 
			\left\{
			\begin{array}{ll}
				-1/\alpha,\text{ if } \alpha \leq \beta \\
				-1/\beta,\text{ if } \alpha \geq \beta
			\end{array}
			\right.
			, \\
			& \lim_{x_i \rightarrow \infty} \frac{\big[U^{(0)}_{x_i}\big]^2U^{(0)}_{x_ix_ix_i}}{x_i \big[U^{(0)}_{x_ix_i}\big]^3} 
			=
			\left\{
			\begin{array}{ll}
				-(\alpha+1)/\alpha^2, \text{ if } \alpha \leq \beta \\
				-(\beta+1)/\beta^2,  \text{ if } \alpha \geq \beta
			\end{array}
			\right.
			, \\
			& \lim_{x_i \rightarrow \infty} \frac{U^{(0)}_{x_ix_i}}{x_i U^{(0)}_{x_ix_ix_i}} = 
			\left\{
			\begin{array}{ll}
				- 1/(\alpha+1), \text{ if } \alpha \leq \beta \\
				- 1/(\beta+1), \text{ if } \alpha \geq \beta
			\end{array}
			\right.
			, \\
			& \lim_{x_i \rightarrow \infty} \frac{\big[U^{(0)}_{x_ix_i}\big]^3}{x_iU^{(0)}_{x_i}\big[U^{(0)}_{x_ix_ix_i}\big]^2} 
			= 
			\left\{
			\begin{array}{ll}
				-\alpha/(\alpha+1)^2,  \text{ if } \alpha \leq \beta \\
				-\beta/(\beta+1)^2,  \text{ if } \alpha \geq \beta
			\end{array}
			\right.
			, \\
			& \lim_{x_i \rightarrow \infty} \frac{U^{(0)}_{x_ix_ix_i}}{x_i U^{(0)}_{x_ix_ix_ix_i}} 
			= 
			\left\{
			\begin{array}{ll}
				-1/(\alpha+2), \text{ if } \alpha \leq \beta \\
				-1/(\beta+2), \text{ if } \alpha \geq \beta
			\end{array}
			\right.
			, \\
			& \lim_{x_i \rightarrow \infty} \frac{\big[U^{(0)}_{x_ix_i}\big]^2}{x_i U^{(0)}_{x_i}U^{(0)}_{x_ix_ix_ix_i}} 
			= 
			\left\{
			\begin{array}{ll}
				-\alpha/(\alpha+1)(\alpha+2), \text{ if } \alpha \leq \beta \\
				-\beta/(\beta+1)(\beta+2), \text{ if } \alpha \geq \beta
			\end{array}
			\right.
			.
		\end{align*}
		Hence, we have $ \big| \sigma_i(t) \widetilde{\pi}_i(t) / x_i \big| \leq c_i $ with some constant $ c_i $, for $ i = 1,\dots,n $. Finally, a similar argument to \cite{Kumar} (Lemma A.1) gives:
		\begin{align*}
			\mathbb{E} \bigg( \int_0^T \sigma_i^2(t) \widetilde{\pi}_i^2(t) dt \bigg)\ +\ 
			\mathbb{E} \bigg( \int_0^T \frac{\sigma^2_i(t) \widetilde{\pi}^2_i(t)}{\widehat{X}^{2r}_i(t)} dt \bigg) < \infty, \quad r>0, \quad i = 1,\dots,n.
		\end{align*}
		By Definition \ref{def_1}, we prove that the optimal portfolio $ \widetilde{\pi}_i(t,\boldsymbol{x},y)  $ is admissible under Assumption \ref{Assume_UT}.
	\end{proof}

\section{Appendix: Proof of Theorem \ref{thm}}
\label{apendixA1}
	
	\begin{proof} We prove this theorem in two steps. \\
		\textbf{Step 1.} We construct super- and subsolutions to HJB equation \eqref{HJB}. We begin by expanding the utility function at terminal time $ T $ using power series:
		\begin{align} \label{U}
			U(t, \boldsymbol{x}, y) =
			U^{(0)}(\boldsymbol{x},y) + (T-t) U^{(1)}(\boldsymbol{x}, y) + (T-t)^2 U^{(2)}(\boldsymbol{x}, y)
			+ O\big((T-t)^3\big),
		\end{align}
		where $ U^{(0)}(\boldsymbol{x}, y) = U(T, \boldsymbol{x}, y) = U_T(\boldsymbol{x}) $ by terminal condition.
		For convenience, we denote $ U^{(0)} = U_T(\boldsymbol{x}) $, $ \delta = T-t $, $ U^{(1)} = U^{(1)}(\boldsymbol{x}, y) $, and $ U^{(2)} = U^{(2)}(\boldsymbol{x}, y) $.
		
		Referring to \cite{Fouque}, \cite{Nadtochiy} and \cite{Kumar}, we substitute the expansion \eqref{U} into HJB equation \eqref{HJB} and optimal portfolio \eqref{pi_hat}, and collect terms in powers of $ \delta $.
		Noting $ \delta = T-t > 0 $,
		and comparing the terms between two sides of the HJB equation, we obtain the expansion of HJB equation \eqref{HJB} up to the highest order $ \delta $:
		\begin{multline} \label{HJB_2}
			-U^{(1)} - 2\delta U^{(2)}
			+ \sum_{i=1}^{n} \sigma_i \widehat{\pi}_i^{(1)}  \lambda_i \Big[U^{(0)}_{x_i} + \delta U^{(1)}_{x_i}\Big] 
			+ \frac{1}{2} \sum_{i,j=1}^{n} \rho_{ij}  \sigma_i \sigma_j \widehat{\pi}_i^{(1)}   \widehat{\pi}_j^{(1)}  \Big[U^{(0)}_{x_i x_j} + \delta U^{(1)}_{x_i x_j}\Big] 
			\\
			+ a \delta \sum_{i=1}^n  \sigma_i \widehat{\pi}_i^{(1)}  \omega_i U^{(1)}_{x_iy}
			+ b \delta U_y^{(1)}
			+\frac{1}{2}a^2 \delta U_{yy}^{(1)} =0,
		\end{multline}
		where the first-order optimal portfolio is
		\begin{multline} \label{pi_2}
			\widehat{\pi}_i^{(1)} 
			= 
			\frac{1}{2 \sigma_i \big[U^{(0)}_{x_ix_i} + \delta U^{(1)}_{x_ix_i}\big]}
			\Bigg[
			- 3\Big[\lambda_i \big[U^{(0)}_{x_i} + \delta U^{(1)}_{x_i}\big] + a \omega_i \delta U^{(1)}_{x_iy}\Big] 
			\\
			+ \sum_{k=1}^n 
			\rho_{ik} \Big[\lambda_k\big[U^{(0)}_{x_k} + \delta U^{(1)}_{x_k}\big]+a \omega_k \delta U^{(1)}_{x_ky}\Big]
			\Bigg].
		\end{multline}
		At the highest order $ \delta^0 $ in \eqref{HJB_2} and \eqref{pi_2}, we obtain
		\begin{align} \label{U_1}
			- U^{(1)} + 
			\sum_{i=1}^{n} \sigma_i \widehat{\pi}_i^{(0)} \lambda_i U^{(0)}_{x_i}
			+ \frac{1}{2} \sum_{i,j=1}^{n} \rho_{ij} \sigma_i \sigma_j \widehat{\pi}_i^{(0)} \widehat{\pi}_j^{(0)} U^{(0)}_{x_i x_j} = 0,
		\end{align}
		where the zero-th order optimal portfolio is
		\begin{align} \label{pi_0}
			\widehat{\pi}_i^{(0)}
			= \frac{1}{2 \sigma_i}
			\bigg[-3\lambda_i  +  \sum_{k=1}^n \rho_{ik} \lambda_k\bigg]
			\frac{U^{(0)}_{x_i}}{ U^{(0)}_{x_ix_i}}.
		\end{align}
		Here we recall $ \| \boldsymbol{x} \| := \sum_{i=1}^{n} x_i $ which implies $ U^{(0)}_{x_i} = U^{(0)}_{x_j} $ and $ U^{(0)}_{x_ix_i} = U^{(0)}_{x_ix_j} $, $ i,j = 1, \dots, n $. Substituting \eqref{pi_0} into \eqref{U_1}, we solve $ U^{(1)} $ as equation \eqref{U^1}.

		We solve the $ U^{(2)} $ by substituting \eqref{pi_2} into \eqref{HJB_2}, clearing fractions, and equating the sum of coefficients of $ \delta $ to zero. Firstly we rewrite the first-order optimal portfolio \eqref{pi_2} as:	
		\begin{align} \label{pi_3}
			\widehat{\pi}_i^{(1)}
			= 
			\frac{A_i + \delta B_i}{2\sigma_i \big[U^{(0)}_{x_ix_i} + \delta U^{(1)}_{x_ix_i}\big]},
		\end{align}
		where 
		\begin{align*}
			& A_i =  - 3\lambda_i U^{(0)}_{x_i} +  \sum_{k=1}^n \rho_{ik} \lambda_kU^{(0)}_{x_k}, \\
			& B_i = - 3 \big[ \lambda_i U^{(1)}_{x_i} + a \omega_i U^{(1)}_{x_iy}\big] + \sum_{k=1}^n 
			\rho_{ik} \big[ \lambda_kU^{(1)}_{x_k}+a \omega_k U^{(1)}_{x_ky}\big].
		\end{align*}
		Then substituting \eqref{pi_3} into \eqref{HJB_2}, and taking reduction of fractions, we obtain the expansion of HJB equation \eqref{HJB} at the highest order $ \delta $:	
		\begin{align*}
			0 = 
			& - U^{(1)} + \delta \Big[ b U_y^{(1)} + \frac{a^2}{2} U_{yy}^{(1)} - 2U^{(2)} \Big]
			+ \frac{1}{2}\sum_{i=1}^{n} \frac{A_i + \delta B_i}{U^{(0)}_{x_ix_i} + \delta U^{(1)}_{x_ix_i}} \cdot \lambda_i \Big[U^{(0)}_{x_i} + \delta U^{(1)}_{x_i}\Big] \\
			& + \frac{1}{8} \sum_{i,j=1}^{n} \frac{A_i + \delta B_i}{U^{(0)}_{x_ix_i} + \delta U^{(1)}_{x_ix_i}}\cdot \frac{A_j + \delta B_j}{ U^{(0)}_{x_jx_j} + \delta U^{(1)}_{x_jx_j}} \cdot \rho_{ij} \Big[U^{(0)}_{x_i x_j} + \delta U^{(1)}_{x_i x_j}\Big]
			+ \delta \frac{a}{2} \sum_{i=1}^n  \frac{A_i + \delta B_i}{U^{(0)}_{x_ix_i} + \delta U^{(1)}_{x_ix_i}} \cdot \omega_i U^{(1)}_{x_iy} \\
			=
			& - U^{(1)} + \delta \Big[ b U_y^{(1)} + \frac{a^2}{2} U_{yy}^{(1)} - 2U^{(2)} \Big] 
			+ \frac{1}{2 \Big[ 1 + \delta \sum_{l=1}^n \frac{U^{(1)}_{x_lx_l}}{U^{(0)}_{x_lx_l}}\Big]} \cdot \\
			& 
			\Bigg[
			\sum_{i=1}^{n} \bigg[ A_iU^{(0)}_{x_i} + \delta \Big[B_iU^{(0)}_{x_i}+A_iU^{(1)}_{x_i}\Big] \bigg]
			\bigg[ 1 + \delta \sum_{\substack{l=1 \\ l \neq i}}^n \frac{U^{(1)}_{x_lx_l}}{U^{(0)}_{x_lx_l}}\bigg] 
			\frac{\lambda_i}{U^{(0)}_{x_ix_i}} \\
			& \
			+ \frac{1}{4}
			\sum_{i,j=1}^{n} \bigg[ A_iA_jU^{(0)}_{x_i x_j} + \delta \Big[A_iB_jU^{(0)}_{x_i x_j} + B_iA_jU^{(0)}_{x_i x_j} + A_iA_jU^{(1)}_{x_i x_j}\Big]\bigg]
			\bigg[ 1 + \delta \sum_{\substack{l=1 \\ l \neq i,j}}^n \frac{U^{(1)}_{x_lx_l}}{U^{(0)}_{x_lx_l}}\bigg] 
			\frac{\rho_{ij}}{U^{(0)}_{x_ix_i}U^{(0)}_{x_jx_j}} \\
			& \
			+ \delta a
			\sum_{i=1}^n \bigg[ A_i + \delta B_i \bigg] 
			\bigg[ 1 + \delta \sum_{\substack{l=1 \\ l \neq i}}^n \frac{U^{(1)}_{x_lx_l}}{U^{(0)}_{x_lx_l}}\bigg] 
			\frac{\omega_i U^{(1)}_{x_iy}}{U^{(0)}_{x_ix_i}}
			\Bigg].
		\end{align*}
		Finally clearing fractions and equating the sum of coefficients of $ \delta $ to zero with simplification solve the $ U^{(2)} $ as
		\begin{align} \label{U^2}
			& U^{(2)} = \nonumber \\
			& - \frac{1}{2} U^{(1)} \sum_{i=1}^n \frac{U^{(1)}_{x_ix_i}}{U^{(0)}_{x_ix_i}} + \frac{b}{2}  U_y^{(1)} + \frac{a^2}{4} U_{yy}^{(1)}
			+ \frac{a}{4} \sum_{i=1}^n  A_i \frac{\omega_i U^{(1)}_{x_iy}}{U^{(0)}_{x_ix_i}} \nonumber \\
			& + \frac{1}{4} 
			\sum_{i=1}^{n} \frac{\lambda_i}{U^{(0)}_{x_ix_i}}  
			\bigg[ 
			B_iU^{(0)}_{x_i}+A_iU^{(1)}_{x_i} + A_iU^{(0)}_{x_i} 
			\sum_{\substack{l=1 \\ l \neq i}}^n \frac{U_{x_lx_l}^{(1)}}{U_{x_lx_l}^{(0)}}
			\bigg] 
			\nonumber \\
			& + \frac{1}{16} 
			\sum_{i,j=1}^{n} \frac{\rho_{ij}}{U^{(0)}_{x_ix_i}U^{(0)}_{x_jx_j}}
			\Bigg[ 
			\bigg[A_iH_j + B_iA_j\bigg]U^{(0)}_{x_i x_j} 
			+ A_iA_j \bigg[ U^{(1)}_{x_i x_j} 
			+ U^{(0)}_{x_i x_j} \sum_{\substack{l=1 \\ l \neq i,j}}^n \frac{U_{x_lx_l}^{(1)}}{U_{x_lx_l}^{(0)}} \bigg]
			\Bigg],
		\end{align}
		where
		\begin{align} \label{U_1p}
			U^{(1)} = \
			& - \frac{1}{2}\sum_{i=1}^{n} \frac{\lambda_i \big[U^{(0)}_{x_i}\big]^2}{U^{(0)}_{x_ix_i}}
			\bigg[3\lambda_i
			- \sum_{k=1}^n \rho_{ik} \lambda_k \bigg] \nonumber\\
			& + \frac{1}{8} \sum_{i,j=1}^{n} \frac{\rho_{ij}\big[U^{(0)}_{x_i}\big]^2}{U^{(0)}_{x_ix_i}}
			\bigg[3\lambda_i - \sum_{k=1}^n \rho_{ik} \lambda_k \bigg] 
			\bigg[3\lambda_j - \sum_{k=1}^n \rho_{jk} \lambda_k \bigg], 
		\end{align}
		\begin{align*} 
			U^{(1)}_{x_i} = \
			&  
			- \frac{1}{2} \sum_{i=1}^{n} \lambda_i 
			U^{(0)}_{x_i} \bigg[2 - \frac{U^{(0)}_{x_i}U^{(0)}_{x_ix_ix_i}}{\big[U^{(0)}_{x_ix_i}\big]^2}\bigg]
			\bigg[3\lambda_i - \sum_{k=1}^n \rho_{ik} \lambda_k\bigg] \\
			& + \frac{1}{8} \sum_{i,j=1}^{n} \rho_{ij} 
			U^{(0)}_{x_i} \bigg[2 - \frac{U^{(0)}_{x_i}U^{(0)}_{x_ix_ix_i}}{\big[U^{(0)}_{x_ix_i}\big]^2}\bigg]
			\bigg[3\lambda_i - \sum_{k=1}^n \rho_{ik} \lambda_k \bigg] 
			\bigg[3\lambda_j - \sum_{k=1}^n \rho_{jk} \lambda_k \bigg],
		\end{align*}
		\begin{align} \label{U_1xx}
			& U^{(1)}_{x_ix_j} = U^{(1)}_{x_ix_i} = \nonumber \\
			& - \frac{1}{2} \sum_{i=1}^{n} \lambda_i 
			U^{(0)}_{x_ix_i} 
			\bigg[2 - \frac{U^{(0)}_{x_i}U^{(0)}_{x_ix_ix_i}}{\big[U^{(0)}_{x_ix_i}\big]^2}\bigg]
			\bigg[3\lambda_i - \sum_{k=1}^n \rho_{ik} \lambda_k\bigg] \nonumber \\
			& + \frac{1}{2} \sum_{i=1}^{n} \lambda_i 
			U^{(0)}_{x_i}
			\bigg[\frac{U^{(0)}_{x_ix_ix_i}}{U^{(0)}_{x_ix_i}} 
			+ \frac{U^{(0)}_{x_i}U^{(0)}_{x_ix_ix_ix_i}}{\big[U^{(0)}_{x_ix_i}\big]^2} 
			- \frac{2U^{(0)}_{x_i}\big[U^{(0)}_{x_ix_ix_i}\big]^2}{\big[U^{(0)}_{x_ix_i}\big]^3}\bigg]
			\bigg[3\lambda_i - \sum_{k=1}^n \rho_{ik} \lambda_k\bigg]  \nonumber \\
			& + \frac{1}{8} \sum_{i,j=1}^{n} \rho_{ij} 
			U^{(0)}_{x_ix_i} 
			\bigg[2 - \frac{U^{(0)}_{x_i}U^{(0)}_{x_ix_ix_i}}{\big[U^{(0)}_{x_ix_i}\big]^2}\bigg]
			\bigg[3\lambda_i - \sum_{k=1}^n \rho_{ik} \lambda_k \bigg] 
			\bigg[3\lambda_j - \sum_{k=1}^n \rho_{jk} \lambda_k \bigg] \nonumber \\
			& - \frac{1}{8} \sum_{i,j=1}^{n} \rho_{ij} 
			U^{(0)}_{x_i}
			\bigg[\frac{U^{(0)}_{x_ix_ix_i}}{U^{(0)}_{x_ix_i}} 
			+ \frac{U^{(0)}_{x_i}U^{(0)}_{x_ix_ix_ix_i}}{\big[U^{(0)}_{x_ix_i}\big]^2} 
			- \frac{2U^{(0)}_{x_i}\big[U^{(0)}_{x_ix_ix_i}\big]^2}{\big[U^{(0)}_{x_ix_i}\big]^3}\bigg]
			\bigg[3\lambda_i - \sum_{k=1}^n \rho_{ik} \lambda_k \bigg] 
			\bigg[3\lambda_j - \sum_{k=1}^n \rho_{jk} \lambda_k \bigg].
		\end{align}
		Here we aslo recall $ \| \boldsymbol{x} \| := \sum_{i=1}^{n} x_i $ which implies $ U^{(0)}_{x_i} = U^{(0)}_{x_j} $, $ U^{(0)}_{x_ix_i} = U^{(0)}_{x_ix_j} $, $ U^{(0)}_{x_ix_ix_i} = U^{(0)}_{x_ix_jx_k} $, and $ U^{(0)}_{x_ix_ix_ix_i} = U^{(0)}_{x_ix_jx_kx_l} $, $ i,j,k,l = 1, \dots, n $. 
		
		We enumerate all the terms of $ U^{(2)} $ as $ u_1^{(2)} , \cdots , u_m^{(2)} $, that is $ U^{(2)} = \sum_{i=1}^m u_i^{(2)} $. We then prove Lemma \ref{A.0} in the Appendix \ref{apendixA} showing $ u_i^{(2)} \sim f(\boldsymbol{x}) $ for $ 1\leq i \leq m $, where $ f(\boldsymbol{x}) = 1 $ under Case 1 of Assumption \ref{Assume_UT} and $ f(\boldsymbol{x}) = \|\boldsymbol{x}\|^{1-\alpha} + \|\boldsymbol{x}\|^{1-\beta} $ under Case 2 of Assumption \ref{Assume_UT}.
		Now we let
		\begin{align*}
		u^{(2)}(\boldsymbol{x},y) = 1 + m \cdot \max_{1 \leq i \leq m} \sup \frac{\big|u_i^{(2)}(\boldsymbol{x},y)\big|}{f(\boldsymbol{x})}.
		\end{align*}
		Then we define super-solution $ \overline{U} = \overline{U}(t,\boldsymbol{x},y) $ and sub-solution $ \underline{U} = \underline{U}(t,\boldsymbol{x},y) $ to HJB equation \eqref{HJB} by 
		\begin{align*}
			\overline{U}(t,\boldsymbol{x},y) & = U_T(\boldsymbol{x}) + (T-t)U^{(1)}(\boldsymbol{x},y) + (T-t)^2u^{(2)}(\boldsymbol{x},y)f(\boldsymbol{x}), \\
			\underline{U}(t,\boldsymbol{x},y) & = U_T(\boldsymbol{x}) + (T-t)U^{(1)}(\boldsymbol{x},y) - (T-t)^2u^{(2)}(\boldsymbol{x},y)f(\boldsymbol{x}).
		\end{align*}
		We note the coefficient of $ U^{(2)} $ is $ - 2(T-t) < 0 $ in \eqref{HJB_2}. After substituting $ \overline{U} $ into the left hand side of HJB equation \eqref{HJB} and clearing fractions, we observe that the coefficient of $ T-t $ is strictly negative by definition of super-solution. On the other hand, substituting $ \underline{U} $ gives that the coefficient of $ T-t $ is strictly positive by definition of sub-solution.
		
		Because $ U^{(2)} $ is solved by equating the sum of coefficients of $ T-t $ to zero, therefore $ u_i^{(2)} \sim f(\boldsymbol{x}) $, $ 1\leq i \leq m $,  imply that the coefficient of $ T-t $ in either $ \overline{U}_t + \mathcal{H}(\overline{U}) $ or $ \underline{U}_t + \mathcal{H}(\underline{U}) $ is in the order of $ f(\boldsymbol{x}) $.
		Referring to the proof of Lemma \ref{A.0}, from \eqref{U_1xx}, we observe that $ U^{(1)}_{x_ix_j} $ is a linear combination of the following terms: 
		\begin{align*}
			U^{(0)}_{x_ix_i}, \
			\frac{U^{(0)}_{x_i}U^{(0)}_{x_ix_ix_i}}{U^{(0)}_{x_ix_i}}, \
			\frac{\big[U^{(0)}_{x_i}\big]^2U^{(0)}_{x_ix_ix_ix_i}}{\big[U^{(0)}_{x_ix_i}\big]^2}, \text{ and }
			\frac{\big[U^{(0)}_{x_i}\big]^2 \big[U^{(0)}_{x_ix_ix_i}\big]^2}{\big[U^{(0)}_{x_ix_i}\big]^3},\ i = 1, \dots, n.
		\end{align*}
		For $ i = 1, \dots, n $, under Case 1 of Assumption \ref{Assume_UT},
		\begin{align*}
			& \lim_{\|\boldsymbol{x}\| \rightarrow \infty} U^{(0)}_{x_ix_i} \big/ U^{(0)}_{x_ix_i} = 1, 
			\qquad \qquad \qquad
			\lim_{\|\boldsymbol{x}\| \rightarrow \infty} \frac{U^{(0)}_{x_i}U^{(0)}_{x_ix_ix_i}}{U^{(0)}_{x_ix_i}} \Big/ U^{(0)}_{x_ix_i} = 2, \\
			& \lim_{\|\boldsymbol{x}\| \rightarrow \infty} \frac{\big[U^{(0)}_{x_i}\big]^2U^{(0)}_{x_ix_ix_ix_i}}{\big[U^{(0)}_{x_ix_i}\big]^2} \Big/ U^{(0)}_{x_ix_i} = 6,
			\quad
			\lim_{\|\boldsymbol{x}\| \rightarrow \infty} \frac{\big[U^{(0)}_{x_i}\big]^2 \big[U^{(0)}_{x_ix_ix_i}\big]^2}{\big[U^{(0)}_{x_ix_i}\big]^3} \Big/ U^{(0)}_{x_ix_i} = 8;
		\end{align*}
		and under Case 2 of Assumption \ref{Assume_UT}, $ \lim_{\|\boldsymbol{x}\| \rightarrow \infty} U^{(0)}_{x_ix_i} \big/ U^{(0)}_{x_ix_i} = 1 $,
		\begin{align*}
			& \lim_{\|\boldsymbol{x}\| \rightarrow \infty} \frac{U^{(0)}_{x_i}U^{(0)}_{x_ix_ix_i}}{U^{(0)}_{x_ix_i}} \Big/ U^{(0)}_{x_ix_i} 
			= 
			\left\{
			\begin{array}{ll}
				1+ 1/\alpha, \text{ if } \alpha \leq \beta \\
				1+ 1/\beta, \text{ if } \alpha \geq \beta
			\end{array}
			\right.
			, \\
			& \lim_{\|\boldsymbol{x}\| \rightarrow \infty} \frac{\big[U^{(0)}_{x_i}\big]^2U^{(0)}_{x_ix_ix_ix_i}}{\big[U^{(0)}_{x_ix_i}\big]^2} \Big/ U^{(0)}_{x_ix_i} 
			=
			\left\{
			\begin{array}{ll}
				(\alpha+1)(\alpha+2)/\alpha^2, \text{ if } \alpha \leq \beta \\
				(\beta+1)(\beta+2)/\beta^2, \text{ if } \alpha \geq \beta
			\end{array}
			\right.
			, \\
			& \lim_{\|\boldsymbol{x}\| \rightarrow \infty} \frac{\big[U^{(0)}_{x_i}\big]^2 \big[U^{(0)}_{x_ix_ix_i}\big]^2}{\big[U^{(0)}_{x_ix_i}\big]^3} \Big/ U^{(0)}_{x_ix_i} 
			= 
			\left\{
			\begin{array}{ll}
				(\alpha +1)^2/\alpha^2, \text{ if } \alpha \leq \beta \\
				(\beta +1)^2/\beta^2, \text{ if } \alpha \geq \beta
			\end{array}
			\right.
			.
		\end{align*}
		Hence, for $ i,j = 1, \dots, n $, we have $ U^{(1)}_{x_ix_j}  \sim U^{(0)}_{x_ix_j} $, then $ U^{(0)}_{x_ix_j} U^{(1)}_{x_ix_j} \sim \big[ U^{(0)}_{x_ix_j} \big]^2 $. 
		This further implies that the coefficient of $ T-t $ in either $ \big[\overline{U}_{x_ix_j}\big]^2 $ or $ \big[\underline{U}_{x_ix_j}\big]^2 $ is in the order of $ \big[ U^{(0)}_{x_ix_j} \big]^2 $. 
		Like inequality (3.8) in \cite{Kumar}, for $ i,j = 1, \dots, n $, we obtain the same inequalities as below
		\begin{align} 
			\big[\overline{U}_{x_ix_j}\big]^2 \big|\overline{U}_t + \mathcal{H}(\overline{U})\big| & \leq c_3 (T-t)\tilde{f}(\boldsymbol{x}), \label{3.8}\\
			\big[\underline{U}_{x_ix_j}\big]^2 \big|\underline{U}_t + \mathcal{H}(\underline{U})\big| & \leq c_4 (T-t)\tilde{f}(\boldsymbol{x}), \nonumber
		\end{align}
		where $ c_3 $ and $ c_4 $ are constants, and $ \tilde{f}(\boldsymbol{x}) \sim \big[U^{(0)}_{x_ix_j}\big]^2 f(\boldsymbol{x}) $.
		Under Case 1 of Assumption \ref{Assume_UT}, we have $ \tilde{f}(\boldsymbol{x}) = \|\boldsymbol{x}\|^{-4} $.
		Under Case 2 of Assumption \ref{Assume_UT}, we have
		\begin{align*}
			\lim_{\|\boldsymbol{x}\| \rightarrow \infty} \frac{\big[U^{(0)}_{x_ix_j}\big]^2}{\|\boldsymbol{x}\|^{-2\alpha-2} + \|\boldsymbol{x}\|^{-2\beta-2}} 
			& = 
			\lim_{\|\boldsymbol{x}\| \rightarrow \infty} \frac{\big[-c_1\alpha \|\boldsymbol{x}\|^{-\alpha-1} - c_2\beta \|\boldsymbol{x}\|^{-\beta-1}\big]^2}{\|\boldsymbol{x}\|^{-2\alpha-2} + \|\boldsymbol{x}\|^{-2\beta-2}} \\
			& = 
			\left\{
			\begin{array}{ll}
				c_1^2 \alpha^2, \qquad\qquad \text{ if } \alpha < \beta \\
				(c_1+c_2)^2\alpha^2/2, \text{if } \alpha = \beta \\
				c_2^2\beta^2, \qquad\qquad \text{ if } \alpha > \beta
			\end{array}
			\right.
			,
		\end{align*}
		that is $ \big[U^{(0)}_{x_ix_j}\big]^2 \sim \|\boldsymbol{x}\|^{-2\alpha-2} + \|\boldsymbol{x}\|^{-2\beta-2} $ which gives $ \tilde{f}(\boldsymbol{x}) = \big[\|\boldsymbol{x}\|^{-2\alpha-2} + \|\boldsymbol{x}\|^{-2\beta-2}\big]\big[\|\boldsymbol{x}\|^{1-\alpha} + \|\boldsymbol{x}\|^{1-\beta}\big] $. 
		For $ i,j = 1, \dots, n $, we observe that $ \frac{1}{c_3}\|\boldsymbol{x}\|^{-2} \leq \overline{U}_{x_ix_j} \leq c_3\|\boldsymbol{x}\|^{-2} < 0 $ under Case 1 of Assumption \ref{Assume_UT} and $ \frac{1}{c_3} \big[\|\boldsymbol{x}\|^{-\alpha-1} + \|\boldsymbol{x}\|^{-\beta-1}\big] \leq \overline{U}_{x_ix_j} \leq c_3\big[\|\boldsymbol{x}\|^{-\alpha-1} + \|\boldsymbol{x}\|^{-\beta-1}\big] <0 $ under Case 2 of Assumption \ref{Assume_UT}, $ i,j = 1, \dots, n $, for some $ c_3 <0 $. 
		Thus $ \big[\overline{U}_{x_ix_j}\big]^2 $ is bounded away from zero. 
		So is the $ \big[\underline{U}_{x_ix_j}\big]^2 $ for some $ c_4 <0 $.
		
		For $ i,j = 1, \dots, n $, we already have the coefficient of $ T-t $ in $ \big[\overline{U}_{x_ix_j}\big]^2 \big[\overline{U}_t + \mathcal{H}(\overline{U})\big] $ is strictly negative, in the order of $ \tilde{f}(\boldsymbol{x}) $, and bounded. In addition, we observe that inequality \eqref{3.8} implies the $ o(T-t )$ terms in $ \big[\overline{U}_{x_ix_j}\big]^2 \big[\overline{U}_t + \mathcal{H}(\overline{U})\big] $ are in the order of $ \tilde{f}(\boldsymbol{x}) $ and bounded. For $ t $ near $ T $, the strictly negative coefficient of $ T-t $ uniformly dominates the $ o(T-t ) $ terms, that is $ \overline{U}_t + \mathcal{H}(\overline{U}) < 0 $. Thus $ \overline{U} $ is the classical super-solution of HJB equation \eqref{HJB}. To prove $ \underline{U} $ is the classical sub-solution of HJB equation \eqref{HJB}, we can use a mirror of this discussion. \\
		\textbf{Step 2.} We prove $ \underline{U}(t,\boldsymbol{x},y) \leq J(t,\boldsymbol{x},y) \leq \overline{U}(t,\boldsymbol{x},y) $.
		For $ i = 1, \dots, n $, we substitute super-solution $ \overline{U} $ into optimal portfolio \eqref{pi_hat} to generate the portfolio
		\begin{align*}
			\overline{\widehat{\pi}}_i = \overline{\widehat{\pi}}_i (t,\boldsymbol{x},y)
			= \frac{1}{2\sigma_i \overline{U}_{x_ix_i}} 
			\Bigg[
			-3\Big[\lambda_i\overline{U}_{x_i}+a \omega_i \overline{U}_{x_iy}\Big]
			+  \sum_{k=1}^n 
			\rho_{ik} \Big[\lambda_k \overline{U}_{x_k}+a \omega_k \overline{U}_{x_ky}\Big]
			\Bigg].
		\end{align*}
		Denoting $ \boldsymbol{X}(t) = \big(X_1(t), \dots, X_n(t)\big) $, we apply the multidimensional It\^{o}'s formula to $ \overline{U} $ and obtain
		\begin{align}
			& \overline{U}\big(T, \boldsymbol{X}(T), Y(T)\big) - \overline{U}(t,\boldsymbol{x},y)  \nonumber\\
			& = \int_t^T \bigg[ \overline{U}_t + \sum_{i=1}^{n}\sigma_i \overline{\widehat{\pi}}_i \lambda_i \overline{U}_{x_i} + b\overline{U}_y
			+ \frac{1}{2} \sum_{i,j=1}^{n} \rho_{ij} \sigma_i \overline{\widehat{\pi}}_i  \sigma_j \overline{\widehat{\pi}}_j \overline{U}_{x_ix_j} + a\sum_{i=1}^{n} \sigma_i \overline{\widehat{\pi}}_i \omega_i  \overline{U}_{x_iy} +\frac{1}{2}a^2\overline{U}_{yy} \bigg] ds \\
			& \quad
			+ \sum_{i=1}^{n} \int_t^T \bigg[ \sigma_i \overline{\widehat{\pi}}_i \overline{U}_{x_i} + a \omega_i \overline{U}_y \bigg] dW_i
			+ \int_t^T \bigg[ a \Big[1-\sum_{i=1}^n \omega_i^2\Big]^{1/2} \overline{U}_y \bigg] dW_0. \label{mart_1}
		\end{align}
		We note the stochastic integrals \eqref{mart_1} are local martingales. Then we let $ \{t_n \}_{n=1}^\infty \subset [t, T] $ be a sequence of stopping times such that $ t_n \leq t_{n+1} $ and $ t_n \rightarrow T $. With replacing $ T $ by $ t_n $, we observe the stochastic integrals \eqref{mart_3} are martingales. 
		\begin{align}
			& \overline{U}\big(t_n, \boldsymbol{X}(t_n), Y(t_n)\big) - \overline{U}(t,\boldsymbol{x},y)  \nonumber\\
			& = \int_t^{t_n} \bigg[ \overline{U}_t + \sum_{i=1}^{n}\sigma_i \overline{\widehat{\pi}}_i \lambda_i \overline{U}_{x_i} + b\overline{U}_y
			+ \frac{1}{2} \sum_{i,j=1}^{n} \rho_{ij} \sigma_i \overline{\widehat{\pi}}_i  \sigma_j \overline{\widehat{\pi}}_j \overline{U}_{x_ix_j} + a\sum_{i=1}^{n} \sigma_i \overline{\widehat{\pi}}_i \omega_i  \overline{U}_{x_iy} +\frac{1}{2}a^2\overline{U}_{yy} \bigg] ds 
			\label{mart_2}\\
			& \quad
			+ \sum_{i=1}^{n} \int_t^{t_n} \bigg[ \sigma_i \overline{\widehat{\pi}}_i \overline{U}_{x_i} + a \omega_i \overline{U}_y \bigg] dW_i
			+ \int_t^{t_n} \bigg[ a \Big[1-\sum_{i=1}^n \omega_i^2\Big]^{1/2} \overline{U}_y \bigg] dW_0. \label{mart_3}
		\end{align}
		Furthermore, we note that the integrand of \eqref{mart_2} is exactly the $ \overline{U}_t + \mathcal{H}(\overline{U}) < 0 $. By the martingale property that conditional expectation of martingale is zero, we have $ \mathbb{E}\big(\overline{U}(t_n, \boldsymbol{X}(t_n), Y(t_n)) - \overline{U} \big| \boldsymbol{x}, y \big) < 0 $, that is
		\begin{align} \label{ineq2}
			 \mathbb{E}\Big(\overline{U}\big(t_n, \boldsymbol{X}(t_n), Y(t_n)\big) \Big| \boldsymbol{x}, y \Big) < \overline{U}.
		\end{align}
		
		By definition of $ \overline{U} $ and triangle inequality, we have
		\begin{align*}
			& \Big| \overline{U}\big(t_n, \boldsymbol{X}(t_n), Y(t_n)\big) \Big| \\
			& = \Big| U_T\big(\boldsymbol{X}(t_n)\big) + (T-t_n)U^{(1)}\big(\boldsymbol{X}(t_n),Y_{t_n}\big) + (T-t_n)^2u^{(2)}\big(\boldsymbol{X}(t_n),Y_{t_n}\big) f\big(\boldsymbol{X}(t_n)\big) \Big| \\
			& \leq \Big| U_T\big(\boldsymbol{X}(t_n)\big) \Big| + \Big| T U^{(1)}\big(\boldsymbol{X}(t_n),Y_{t_n}\big) \Big| + \Big| T^2u^{(2)}\big(\boldsymbol{X}(t_n),Y_{t_n}\big) f\big(\boldsymbol{X}(t_n)\big) \Big| \\
			& \leq c_5 g\big(\boldsymbol{X}(t_n)\big),
		\end{align*}
		with some constant $ c_5 $. 
		Referring to the proof of Lemma \ref{A.0}, from \eqref{U_1p}, we observe that $ U^{(1)} $ is a linear combination of the terms that are equivalent to $ [U^{(0)}_{x_i}]^2 \big/ U^{(0)}_{x_ix_i} $, $ i = 1, \dots, n $, and under Case 1 of Assumption \ref{Assume_UT}, $ U^{(1)} \sim f(\boldsymbol{x}) $.
		We then observe $ u_j^{(2)} \sim f(\boldsymbol{x}) \Rightarrow u^{(2)}(\boldsymbol{x},y) f(\boldsymbol{x}) \sim f(\boldsymbol{x}) $. Thus we have $ g(\boldsymbol{x}) = U_T(\boldsymbol{x}) + f(\boldsymbol{x}) = \ln\|\boldsymbol{x}\| +1 $ under this case.
		Under Case 2 of Assumption \ref{Assume_UT}, we observe that $ U_T(\boldsymbol{x}) \sim f(\boldsymbol{x}) $ which gives $ g(\boldsymbol{x}) = f(\boldsymbol{x}) = \|\boldsymbol{x}\|^{1-\alpha} + \|\boldsymbol{x}\|^{1-\beta} $.
		
		We prove Lemma \ref{A.1} in the Appendix \ref{apendixA} showing that $ \big\{ g\big(\boldsymbol{X}(t_n)\big) \big\}_{n=1}^\infty $ is uniformly bounded by an integrable random variable. We already have $ \big| \overline{U}\big(t_n, \boldsymbol{X}(t_n), Y_{t_n}\big) \big| \leq c_5 g\big(\boldsymbol{X}(t_n)\big) $. And we observe that 
		\begin{align*}
			t_n \rightarrow T\ \Rightarrow\ \overline{U}\big(t_n, \boldsymbol{X}(t_n), Y(t_n)\big) \rightarrow \overline{U}\big(T, \boldsymbol{X}(T), Y(T)\big) = U_T\big(\boldsymbol{X}(T)\big).
		\end{align*}
		By the dominated convergence theorem, we have
		\begin{align} \label{lim}
			\lim_{n \rightarrow \infty} \mathbb{E}\Big(\overline{U}\big(t_n, \boldsymbol{X}(t_n), Y(t_n)\big) \Big| \boldsymbol{x}, y \Big) = \mathbb{E}\Big(U_T\big(\boldsymbol{X}(T)\big) \Big| \boldsymbol{x}, y \Big).
		\end{align}
		Combining \eqref{ineq2} with \eqref{lim}, we have $ \mathbb{E}\big(U_T(\boldsymbol{X}(T)) \big| \boldsymbol{x}, y \big) < \overline{U} $ . 
		We also prove Lemma \ref{A.2} in the Appendix \ref{apendixA} showing that $ \overline{\widehat{\pi}}_i $ is admissible. Thus $ J \leq \overline{U} $. 
		
		To prove $ J \geq \underline{U} $, we can use a mirror of above discussions. Finally, by definitions of $ \underline{U} $ and $ \overline{U} $, the inequality \eqref{ineq} holds.
	\end{proof}

\section{Appendix: Tables for Assumption \ref{Conv}}
\label{apendixC}

\begin{table}[H]
	\caption{Example of portfolios that satisfy the Assumption \ref{Conv}, in different Global Industry Classification Standard (GICS) sectors. The calculations of $ \sum_{j=1}^n |\rho_{ij}| \sigma_j(y) $ and $ 3\sigma_i(y) $ are based on the daily adjusted closing price, from 2023-05-16 to 2024-05-15.}
	\begin{center}
		\begin{tabular}{llcc}
			\hline
			Sector & Company (asset) & $ \displaystyle{\sum_{j=1}^n |\rho_{ij}| \sigma_j(y)} $ & $ 3\sigma_i(y) $ \\
			\hline
			\multirow{3}{*}{Healthcare}
			& Thermo Fisher Scientific Inc. & $ 56.32 $ & $ 113.82 $ \\
			& Amgen Inc. & $ 50.04 $ & $ 83.56 $ \\
			& UnitedHealth Group Incorporated & $ 55.82 $ & $ 76.34 $ \\
			\hline
			\multirow{4}{*}{Real Estate}
			& Digital Realty Trust, Inc. & $ 38.45 $ & $ 44.36 $ \\
			& Simon Property Group, Inc. & $ 46.57 $ & $ 55.02 $ \\
			& Public Storage & $ 34.84 $ & $ 43.74 $ \\
			& American Tower Corporation & $ 39.95 $ & $ 42.66 $ \\
			\hline
			\multirow{5}{*}{Industrials}
			& General Electric Company & $ 59.36 $ & $ 74.98 $ \\
			& Deere \& Company & $ 33.84 $ & $ 61.87 $ \\
			& Lockheed Martin Corporation & $ 33.53 $ & $ 43.90 $ \\
			& Union Pacific Corporation & $ 42.13 $ & $ 56.93 $ \\
			& The Boeing Company & $ 38.98 $ & $ 67.56 $ \\
			\hline
		\end{tabular}
	\end{center}
\end{table}

\begin{table}[H]
	\caption{Example of portfolios that satisfy the Assumption \ref{Conv}, in different GICS sectors. The calculations of $ \sum_{j=1}^n |\rho_{ij}| \sigma_j(y) $ and $ 3\sigma_i(y) $ are based on the daily adjusted closing price, from 2023-11-14 to 2024-05-15.}
	\begin{center}
		\begin{tabular}{llcc}
			\hline
			Sector & Company (asset) & $ \displaystyle{\sum_{j=1}^n |\rho_{ij}| \sigma_j(y)} $ & $ 3\sigma_i(y) $ \\
			\hline
			\multirow{3}{*}{Healthcare}
			& Thermo Fisher Scientific Inc. & $ 58.00 $ & $ 102.71 $ \\
			& Amgen Inc. & $ 29.52 $ & $ 49.29 $ \\
			& UnitedHealth Group Incorporated & $ 53.23 $ & $ 78.14 $ \\
			\hline
			\multirow{4}{*}{Real Estate}
			& Digital Realty Trust, Inc. & $ 14.99 $ & $ 15.08 $ \\
			& Simon Property Group, Inc. & $ 18.94 $ & $ 28.64 $ \\
			& Prologis, Inc. & $ 21.88 $ & $ 31.02 $ \\
			& American Tower Corporation & $ 21.37 $ & $ 34.98 $ \\
			\hline
			\multirow{4}{*}{Industrials}
			& General Electric Company & $ 57.44 $ & $ 74.22 $ \\
			& Lockheed Martin Corporation & $ 34.47 $ & $ 40.78 $ \\
			& Union Pacific Corporation & $ 25.40 $ & $ 27.98 $ \\
			& The Boeing Company & $ 52.89 $ & $ 82.29 $ \\
			\hline
		\end{tabular}
	\end{center}
\end{table}


\textbf{Data availability statement}: The raw data required to reproduce the above findings are available to download from \url{https://finance.yahoo.com}. \\

\noindent \textbf{Acknowledgments}: We thank the anonymous referees for their constructive comments and suggestions, which helped us improve the manuscript.


\end{document}